\documentclass[11pt]{article}

\usepackage{amsthm}
\usepackage{algorithm,algorithmic}
\usepackage{graphicx}
\usepackage{amssymb}
\usepackage{setspace}

\def\denseformat{
	\setlength{\textheight}{9in}
	\setlength{\textwidth}{6.9in}
	\setlength{\evensidemargin}{-0.2in}
	\setlength{\oddsidemargin}{-0.2in}
	\setlength{\headsep}{10pt}
	\setlength{\topmargin}{-0.3in}
	\setlength{\columnsep}{0.375in}
	\setlength{\itemsep}{0pt}
	\renewcommand{\baselinestretch}{1.1}
}

\newtheorem{thm}{Theorem}[section]
\theoremstyle{definition}

\renewcommand{\baselinestretch}{2}
\theoremstyle{plain}
\newtheorem{lem}[thm]{Lemma}
\newtheorem{col}[thm]{Corollary}

\date{}

\denseformat
\begin{document}

\title{Speedup of Distributed Algorithms for Power Graphs\\ in the CONGEST Model}
\author{Leonid Barenboim\thanks{Open University of Israel. E-mail: {\tt leonidb@openu.ac.il}.\newline
$^{**}$ Ben-Gurion University of the Negev. E-mail: {\tt goldeuri@post.bgu.ac.il}} \and Uri Goldenberg$^{**}$}
\maketitle

\begin{abstract}
We obtain improved distributed algorithms in the CONGEST message-passing setting for problems on power graphs of an input graph $G$. This includes Coloring, Maximal Independent Set, and related problems. For $R = f(\Delta^k,n)$, we develop a general deterministic technique that transforms $R$-round LOCAL model algorithms for $G^k$ with certain properties into $O(R \cdot \Delta^{k/2 - 1})$-round CONGEST algorithms for $G^k$. This improves the previously-known running time for such transformation, which was $O(R \cdot \Delta^{k - 1})$. Consequently, for problems that can be solved by algorithms with the required properties and within polylogarithmic number of rounds, we obtain {\em quadratic} improvement for $G^k$ and {\em exponential} improvement for $G^2$. We also obtain significant improvements for problems with larger number of rounds in $G$.
Notable implications of our technique are the following deterministic distributed algorithms:
\begin{itemize}
  \item We devise a distributed algorithm for $O(\Delta^4)$-coloring of $G^2$ whose number of rounds is $O(\log \Delta + \log^{*}n )$.  
This improves exponentially (in terms of $\Delta$) the best previously-known deterministic result of Halldorsson, Kuhn and Maus.\cite{HKM2020} that required $O(\Delta + \log^{*}n)$ rounds, and the standard simulation of Linial \cite{L92} algorithm in $G^k$ that required $O(\Delta \cdot \log^* n)$ rounds. 
  \item  We devise an algorithm for $O(\Delta^2)$-coloring of $G^2$ with $O(\Delta \cdot \log \Delta + \log^*n)$ rounds, and $(\Delta^2 + 1)$-coloring with $O(\Delta^{1.5} \cdot \log \Delta + \log^*n)$ rounds. This improves quadratically, and by a power of $4/3$, respectively, the best previously-known results of Halldorsson, Khun and Maus. \cite{HKM2020}.
  \item For $k> 2$, our running time for $O(\Delta^{2k})$-coloring of $G^k$ is  $O(k \cdot \Delta^{k/2 - 1} \cdot \log \Delta \cdot \log^* n)$. \\ Our running time for $O(\Delta^k)$-coloring of $G^k$ is  $\tilde{O}(k \cdot \Delta^{k - 1} \cdot \log^* n)$. \\ This improves best previously-known results quadratically, and by a power of $3/2$, respectively.
  \item For constant $k > 2$, our upper bound for $O(\Delta^{2k})$-coloring of $G^k$ nearly matches the lower bound of Fraigniaud, Halldorsson and Nolin. \cite{FHN2020} for {\em checking} the correctness of a coloring in $G^k$.

\end{itemize}
\end{abstract}

\newpage
\section{Introduction}

\subsection{Model and Results}
In the distributed message-passing model a communication network is represented by an unweighted $n$-vertex graph $G = (V,E)$ with maximum degree $\Delta$. Each vertex has a unique ID, represented by $O(\log n)$ bits. Computations proceed in synchronous discrete rounds, each of which consists of message exchange between neighbors, and local computations of vertices. The input for an algorithm is the network graph $G$, where initially each vertex knows only its own ID and the IDs of its neighbors\footnote{A node can learn the ids of its neighbors within a single round in the CONGEST model}, and the values $n$ and $\Delta$. During an execution, within $k > 0$ rounds, vertices may obtain information about other vertices in their $k$-hop-neighborhood. The output of an algorithm consists of the final answers returned by all vertices. The running time of an algorithm is the number of rounds until the algorithm terminates in all vertices. In the current paper we focus on the {\em CONGEST} model. In this model the number of bits that can be sent over each edge in each round is bounded by $O(\log n)$. Consequently, within $k$ rounds each vertex can learn only a small portion of the information that resides in its $k$-hop-neighborhood. Such a neighborhood may contain up to $\Delta^k + 1$ vertices, whose information is much larger than what can be received in $k$ rounds, namely up to $O(k \cdot \Delta \cdot \log n)$ bits. This is in contrast to the distributed {\em LOCAL} model, where message size is unbounded, and each vertex can learn the entire $k$-hop-neighborhood within $k$ rounds.

Among the most studied problems in this setting are coloring and maximal independent set (henceforth, MIS). These problems are very well motivated by real-life network tasks, such as resource allocation, scheduling, channel assignment, etc. Often, in order to perform such a task, a coloring or MIS has to be computed on a {\em power graph} of the network, rather than the original graph that represents the network. For example, in job scheduling, where each vertex can perform one job at a time, and can send one neighbor a job to perform, a coloring is used for job scheduling. In this scheduling all vertices of a certain color are executed at the same time. Then these vertices and neighbours selected by them execute jobs. However, an ordinary proper coloring will not suffice, since a vertex may have several neighbours with the same color, who send it jobs to perform. Since it can handle only one job, the other jobs that are sent simultaneously are lost. To prevent this, $2$-distance coloring is used, where each pair of vertices at distance at most $2$ one from another in $G$ obtain distinct colors. (This is equivalent to an ordinary proper coloring of the power graph $G^2$.) Now each vertex has at most one neighbor with a certain color, and no more than one job arrives to the vertex at a time.

Because of their importance, problems on power graphs for computing coloring, MIS, and related tasks have been very intensively studied in the distributed setting. A plethora of significant results have been obtained in recent years in the CONGEST setting \cite{BCMPP2020,EPSW2014,FHN2020,HKM2020,HKMN2020,MPU2023}. In particular, Bar-Yehuda, Censor-Hillel, Maus, Pai and Pemmaraju \cite{BCMPP2020} devised approximate Minimum Dominating Set and Minimum Vertex Cover (MVC) algorithms for $G^2$.  Halldorsson, Khun, Maus and Nolin \cite{HKMN2020} devised a logarithmic-time randomized algorithm for distance-$2$ coloring using $\Delta^2 + 1$ colors. (For a positive integer $q$, distance-$2$ coloring with $q$ colors is equivalent to ordinary proper coloring of $G^2$ with $q$ colors.)  Halldorsson, Khun and Maus \cite{HKM2020} devised deterministic algorithms for distance-$2$ coloring using $O(\Delta^2)$ colors with time $\mbox{polylog}(n)$, and using $(\Delta^2 + 1)$ colors with time $O(\Delta^2 + \log^* n)$.  Fraigniaud, Halldorsson and Nolin \cite{FHN2020} showed that for $k > 2$, testing whether a given proper coloring is correct requires $\Omega(\Delta^{\lfloor(k - 1)/2\rfloor})$ rounds. Also, a general well-known scheme for simulating LOCAL $R$-round algorithms for $G^k$ in the CONGEST setting provides algorithms with time $O(R \cdot \Delta^{k-1})$ in the CONGEST setting. 

In the last decades, much attention of researchers was devoted to understanding the complexity of these problems on $G$ and $G^k$ as a function of $\Delta$, modulo the unavoidable factor of $\log^* n$. A major question that remained open for many years is whether $o(\Delta) + \log^* n$ solutions are possible for $O(\Delta)$-coloring, MIS and related problems on $G$. In recent breakthroughs it was shown that while MIS and Maximal Matching require $\Omega(\Delta)$ time \cite{BBHORS19}, the problem of $O(\Delta)$-coloring can be solved in $\tilde{O}(\sqrt{\Delta} + \log^* n)$ time \cite{B15}. However, this question for coloring power graphs in sublinear-in-$\Delta^k$ time still remained open, since the best results for $O(\Delta^k)$-coloring of $G^k$ were at least $O(\Delta^k)$. Specifically, for $k  = 2$ it is $O(\Delta^2 + \log^*n)$ \cite{HKM2020}, and for $k > 2$ it is $\tilde{O}(\Delta^{k + k/2 - 1} + \Delta^{k-1} \log^* n)$ (by applying a standard simulation for $G^k$ to \cite{BEG18}). 

In this paper we answer this question in the affirmative, by providing deterministic CONGEST algorithms for $O(\Delta^2)$-coloring of $G^2$ with $\tilde{O}(\Delta + \log^* n)$ rounds, and $O(\Delta^k)$-coloring of $G^k$ with   $\tilde{O}(\Delta^{k-1} + \log^*n)$ rounds\footnote{For simplicity of presentation, we assume that $k$ is even. Our results extend directly to any positive integer $k \geq 2$, by replacing $k/2$ with $\lceil k/2  \rceil$.}. More generally, we provide a speedup technique for various problems, including coloring and MIS, that improves quadratically each phase of an algorithm for $G^k$ that adheres to certain requirements. In particular, such algorithms that perform the standard simulation are improved from running time  $O(R \cdot \Delta^{k-1})$ to  $O(R \cdot \Delta^{k/2 - 1})$. In the case of $G^2$ our speedup of a phase is {\em exponential} (in terms of $\Delta$), resulting in running time $O(R \cdot \mbox{polylog}(\Delta))$. For example, we compute $O(\Delta^4)$-coloring of $G^2$ within $O(\log \Delta + \log^* n)$ time, improving the best previously-known result of $O(\Delta + \log^* n)$.

Our results also give rise to a quadratic improvement in the {\em memory complexity} per vertex.
Specifically, when using {\em aggregation functions}, the size of the result computed by  $w$ for each vertex in its $k/2$-hop-neighborhood is reduced from $\Delta^{k/2}$ to a much smaller value, ideally, $O(1)$. Using such technique the required memory per vertex is also reduced from $\tilde{O}(\Delta^k)$ to $\tilde{O}(\Delta^{k/2})$.

An interesting implication of our results is that $O(\Delta^{2k})$-coloring of $G^k$ can be computed in the CONGEST model in $O(\Delta^{k/2 - 1} \log \Delta \log^* n)$ rounds. This nearly matches the lower bound of  Fraigniaud, Halldorsson and Nolin \cite{FHN2020} for testing a proper coloring that requires $\Omega(\Delta^{\lfloor(k - 1)/2\rfloor})$ rounds.


\subsection{Our Techniques}
The previously-known technique for simulating CONGEST algorithms from $G$ on $G^k$ proceeds as follows. For each round of the simulated algorithm, each vertex $v \in V$ has to obtain its $k$-hop-neighborhood information. Since the number of vertices in the $k$-hop-neighborhood is bounded by $O(\Delta^k)$ and the number of edges is bounded by $O(\Delta^k \cdot \Delta) = O(\Delta^{k+1}$), this neighborhood structure (consisting of vertices and edges between them) can be delivered to $v$ within $O(\Delta^{k})$ rounds. This information is sufficient for $v$ to simulate its local computation in that round for the algorithm in $G^k$. Then, $v$ broadcasts a message to all vertices in its $k$-hop-neighborhood. Such a broadcast is performed by all vertices of $G$ in parallel, and requires $O(\Delta^{k - 1})$ rounds. We note that often it is sufficient to employ $O(\Delta^{k-1})$ rounds also in the stage of obtaining the $k$-hop-neighborhood information. This is the case when only information about vertices is needed, rather than how they are connected. Consequently, various CONGEST algorithms that require $f(\Delta,n)$ rounds in $G$ can be transformed into $O(f(\Delta^k,n) \cdot \Delta^{k-1})$-round algorithms for $G^k$.

Our new method improves this idea, by performing stages of information collection and information broadcast only half the way, to distance $k/2$ rather than $k$. Indeed, for each pair of vertices $u,v \in V$ at distance $k$ one from another, there is a vertex $w$ in the middle of a path between them that can obtain their information by collecting its $k/2$-hop-neighborhood. Then $w$ can perform a computation on $v$ and $w$ (as well as all other vertices in its $k/2$-hop-neighborhood) and return the results to them. However, if this is done in a trivial way, then $w$ should return information of size $\Delta^{k/2}$ to each vertex in its $k/2$-hop-neighborhood. In general, this is done not only by $w$, but by all vertices in the graph in parallel. Consequently, in order to collect $\Delta^{k/2}$ pieces of information from each of $\Delta^{k/2}$ vertices in a $k/2$-hop-neighborhood, again $O(\Delta^k)$ rounds are needed. Thus, a straightforward approach for going half the way does not provide an improvement. But we obtain an improvement using the following more sophisticated method.

Various algorithms can be decomposed into basic steps that perform such operations as checking whether a variable of a vertex appears in its $k$-hop-neighborhood, summing the number of certain value, computing maximum or minimum, etc. In such cases, rather than collecting the entire $k$-hop-neighborhood information and then computing a function locally, aggregation functions can be applied iteratively and distributively. For example, consider a function for computing the maximum of variables in the $k$-hop-neighborhood. In this case, $w$ who is in the middle of a $k$-length path between $u$ and $v$ applies the function iteratively and locally on all vertices in the $k/2$-hop-neighborhood of $w$. The vertex $w$ sends the result to all its immediate neighbors, including the one that is at distance $k/2-1$ from $v$. Then this vertex computes the maximum of the maxima it received from its neighbors, and sends it to its own neighbors, one of which is at distance $k/2-2$ from $v$. After $k/2$ such rounds, $v$ receive the maximum value in its $k$-hop-neighborhood.

The example of the maximum function is a simple case, when the maximum in the $k/2$-hop-neighborhood of $w$ is the same for all its vertices. However, in general, different answers may be needed for different vertices. Another example is a function that checks whether variables of vertices are equal. It may be the case that $u$ has a vertex in the $k/2$-hop-neighborhood of $w$ with a variable that equals to that of $u$, but $v$ does not have such a variable. Then, $w$ must store for each vertex in its $k/2$-hop-neighborhood its own answer. This is done by $w$, by applying aggregation functions locally and iteratively, for each such vertex.
The result of size $\Delta^{k/2}$ is sent to all neighbours of $w$ in parallel. Then, each neighbor applies the aggregation functions for all vertices in its $(k/2-1)$-hop-neighborhood. It does so starting from the results it received from $w$. This way, the outcome now is regarding distance $k/2 + 1$. This continues for $k/2$ phases, where in each phase less data has to be sent (it is reduced by a factor of $\Delta$ in each phase), but the radius  of the computation grows by 1. After $k/2$ such phases, each vertex holds the result of function applications in its $k$-hop-neighborhood.

In order for this technique to work, it employs functions that are (1) {\em commutative} and (2) {\em idempotent}. That is, (1) the order of function application must not affect the result, and (2) applying the same function several times must not affect the result, no matter how many times it is applied. For example, this is the case in the function $max_t(x) = max(x,t)$. A series of applications $max_{t_1}(max_{t_2}(...(max_{t_q}(x)...))$ can be applied in any order, and each $max_{t_i}$ can appear any positive number of times, without changing the outcome. This is important to the success of our method, since pairs of vertices $u,v$ may belong to many $k/2$-hop-neighborhoods that perform computations for them, and the order of computations is not predefined. As mentioned above, we show  that various algorithms for complicated tasks, including coloring and MIS, can be constructed from steps consisting of such operations.

We extend our technique also to functions that are not idempotent, in order to generalize it further. In particular, the counting operation is a very useful building block in various algorithms. For example, it makes it possible to compute how many times a certain value appears in the variables of a $k$-hop-neighborhood of a vertex. (The function increases the result by one each time it is applied on a variable with that value. But when the function is applied on a variable with a different value, the result does not change.) However, this operation is not idempotent, since the outcome depends on the number of times the function is applied. Consequently, if we apply our technique as described above, the result may be larger than the actual number of variables with that value. This is because a certain function invocation may be repeated several times by different vertices in the $k$-hop-neighborhood. 

We propose two solutions for this challenge. In the first solution we analyze how much the result over-counts the correct answer. In certain cases this can be bounded, so an algorithm still works, even with over-counting. In more complicated cases, when precise computation is required, we use BFS trees of radius $k/2$ that are constructed from all vertices in parallel. These trees are broadcasted to distance $k/2$. Then unique paths can be produced between pairs of vertices at distance $k$ one from another. These paths are used to make sure that each function application is executed exactly once, for an (ordered) pair of vertices in a $k$-hop-neighborhood. This incurs an additional running time of $O(\Delta^{k-1})$, but only once during execution. Consequently, the running time of a transformation of an $f(\Delta,n)$-round algorithm for $G$ becomes $O(\Delta^{k-1} + f(\Delta^k,n) \cdot \Delta^{k/2 - 1})$ in $G^k$. This is again a significant improvement over the previously known time of $O(f(\Delta^k, n) \cdot \Delta^{k-1})$.

Another tool we introduce for shrinking message size, which may be of independent interest, is {\em binary search in neighborhoods}. A common building block of distributed algorithms is performing computations on lists that vertices hold, as follows. A vertex has to compute a certain function on its own list and the lists of its neighbors. In some functions the outcome can be determined by a single element of a list. (For example, finding an element that does not appear in the neighbors lists.) 
A naive computation for lists of size $t$ requires $t$ rounds, since the lists have to be delivered to neighbors that apply the function on them. To speedup this process we perform a binary search with assistance of neighbors, so that lists shrink by a certain factor in each round, but each of them still contains an element from which the function outcome can be deduced. Finally, each list contains just a single element, who provides the desired result. This tool is a main ingredient in our algorithm for $O(\Delta^4)$-coloring of $G^2$ in $O(\log \Delta + \log^* n)$ time, and $O(\Delta^2)$-coloring of $G^2$ in $O(\Delta \cdot \log \Delta +  \log^*n)$ time.
\subsection{Related Work}
Among the first works on deterministic distributed symmetry breaking are algorithms for (1-hop) coloring and MIS of paths and trees. An $O(\log^* n)$-round algorithm for $3$-coloring paths was devised by Cole and Vishkin in 1986 \cite{CV86}. This was extended to oriented trees by Goldberg, Plotkin, and Shannon in 1987 \cite{GPS88}. A (1-hop) coloring deterministic algorithm for general graphs that employs $O(\Delta^2)$-colors and has running time $O(\log^* n)$ was obtained by Linial in 1987 \cite{L92}. This algorithm gives rise to $(\Delta + 1)$-coloring in $O(\Delta^2 + \log^* n)$ time. The running time for $(\Delta + 1)$-coloring was improved to $O(\Delta \cdot \log \Delta + \log^* n)$ by Szegedy and Vishwanathan in 1993 \cite{SV93}, and by Kuhn and Wattenhofer in 2006 \cite{KW06}, by a more explicit construction. This was further improved to $O(\Delta + \log^*)$ by Barenboim, Elkin and Kuhn in 2009 \cite{BE09}. The latter result also implies an algorithm for MIS in $O(\Delta + \log^*n)$ rounds. As proven in \cite{BBHORS19}, this result for MIS is tight. On the other hand, sublinear-in-$\Delta$ algorithms for $(\Delta + 1)$-coloring are possible, as shown by Barenboim \cite{B15}, who devised such an algorithm with running time $\tilde{O}(\Delta^{3/4} + \log^*n)$. This was further improved in several works \cite{BEG18,FHK16,FK23,M21}, to the current state of the art, which is $\tilde{O}(\sqrt{\Delta} + \log^* n)$.

In addition to the thread of research on algorithms with dependency on $\Delta$ in the running time, there has been also progress with deterministic algorithms that depend on $n$ (in a stronger way than just $\log^* n$). Several results were obtained using network decompositions \cite{AGLP1989, PS1996, GGR21}, and the recent breakthrough of Rozhon and Ghaffari \cite{GGR21} makes it possible to compute coloring and MIS deterministically within $\mbox{poly}(\log n)$ time. In addition, it is even possible to obtain $(\Delta+1)$-coloring and MIS without network decompositions in $O(\log n \log^2 \Delta)$ time \cite{GK2022, FGGKR2023} and $(\log^2 n)$ for MIS \cite{GG2023}. Both threads of research 
are very important and attracted much attention of researchers. In the case where the dependency on $n$ is larger than $\log^* n$, a main goal is improving this dependency, as well as the dependency on $\Delta$. In the case that the dependency on $n$ is $O(\log^* n)$, which is unavoidable, the main goal is improving the dependency on $\Delta$. 

Randomized symmetry breaking algorithms have been very extensively studied as well. The first algorithms for $(\Delta + 1)$-coloring and MIS, due to Luby \cite{L86} and Alon, Babai and Itai \cite{ABI86} required $O(\log n)$-time. This was improved by Kothapalli, Scheideler, Onus and Schindelhauer in 2006 \cite{KSOS06} who obtained $O(\Delta)$-coloring in $\tilde{O}(\sqrt {\log n})$ rounds. This was further improved by Barenboim and Elkin \cite{BE10} who obtained algorithms with running time  $O(\mbox{polylog}(\Delta) + 2^{O(\sqrt {\log \log n})})$  for various symmetry breaking problems, including coloring and MIS. In several major advances, this was improved even further to $O(\log \Delta + \mbox{poly}(\log \log n))$ for MIS  \cite{G16,RG20}, $O(\sqrt{\log \Delta} + \mbox{poly}(\log \log n))$ for coloring \cite{HSS16}, and $O(\mbox{poly}(\log \log n))$ for coloring \cite{CLP20}. 

Since currently-known randomized algorithms have better dependency on $\Delta$, while deterministic algorithms have better dependency on $n$, the improvement of either randomized or deterministic solutions is valuable. In particular, improving the dependency on $\Delta$ in deterministic algorithms is very important, since the gap between the current deterministic and randomized solutions is quite large.

Recently, distance-$k$ problems and computations on power graphs attracted much attention in the research of distributed algorithms. Emek, Pfister, Seidel and Wattenhofer \cite{EPSW2014} proved that every problem that
can be solved (and verified) by a randomized anonymous algorithm can also be solved by a deterministic anonymous
algorithm provided that the latter is equipped with a distance-$2$ coloring of the input graph. 
Computing distance-$k$ coloring is a key component in the derandomization of LOCAL distributed algorithms, due to Ghaffari, Harris and Kuhn \cite{GHK18} from 2019. Upper- and lower-bounds for approximate Minimum Dominating Set on power graphs were devised by Bar-Yehuda, Censor-Hillel, Maus, Pai and Pemmaraju \cite{BCMPP2020} in 2020.  Deterministic and randomized distance-$2$ coloring algorithms were obtained by Halldorsson, Kuhn and Maus \cite{HKM2020}. Improved randomized results for distance-$2$ coloring were obtained by Halldorsson, Kuhn, Maus and Nolin \cite{HKMN2020}. Very recently, in 2023, Maus, Peltonen, and Uitto \cite{MPU2023} devised deterministic algorithms for $k$-ruling sets on $G^k$ with time $\tilde{O}(k^2 \log^4 n \log \Delta)$. They also devised randomized algorithm for this problem, as well as for MIS, with logarithm dependency on $\Delta$ and poly-log-log dependency on $n$. The most recent result for 2-distance coloring is a randomized algorithm by Flin, Halldorson and Nolin \cite{FHN2023}, whose running time is $O(\log ^{6} \log n)$.

\section{Distance-$2$ coloring $G$ with $O(\Delta^4)$ colors in $O(\log \Delta + \log^* n)$ rounds}\label{Linal2Distance}
In this section we devise an algorithm for distance-2 coloring of $G$ using $O(\Delta^4)$ colors, which is a distance-1 coloring of $G^2$. Our algorithm significantly speeds-up the previously-known algorithms for distance-2 coloring with this number of colors. The previous algorithms \cite{HKM2020} are based on simulating Linial's \cite{L92} algorithm in $G^2$. (See also \cite{BE20}, for more details about the original algorithm of Linial \cite{L92}.) The algorithm of \cite{HKM2020}  for $G^2$ requires spending $\Delta$ rounds to simulate certain rounds of Linial's algorithm. On the other hand, our new method improves this, so that only $O(\log \Delta)$ rounds are required to simulate a round of Linial's algorithm. Moreover, just a constant number of rounds of Linial's algorithm need to be simulated with $O(\log \Delta)$ rounds, and the other rounds can be simulated with $O(1)$ rounds each.
Consequently, our algorithm produces a proper $O(\Delta^4)$ coloring of $G^2$ within $O(\log \Delta + \log^* n)$ rounds in the CONGEST model. 

The principle of our method is binary search. Consider the $i$-th round, $i = 1,2,...,O(\log^* n)$, of Linial's original algorithm. Each vertex $v$ generates a subset of possible colors
$S(\phi(v)) = \{s_1, s_2, ... , s_j\}$ from the palette of  $\{1,2,... (\Delta \cdot \log^{(i)} n)^2\}$, such that there exists a color $s_j\in S(\phi(v))$, where $s_j \notin S(\phi(u))$, for all neighbors $u$ of $v$.
Moreover, Linial's algorithm makes it possible to construct a set system, such that for any pair of neighbors $u$,$v$, it holds that

$\frac{ |S(\phi(u))| }{ |S(\phi(u)) \cap S(\phi(v))|}   > \Delta$
and
$\frac{ |S(\phi(v))| }{ |S(\phi(u)) \cap S(\phi(v))|}  > \Delta$.

 Each vertex $v$ selects such a color $s_j$, which is from its own set $S(\phi(v))$, but does not belong to any of its neighbors sets, from which the neighbors select colors. Consequently, the coloring is proper. Since in each round the subsets are taken from smaller sets, the number of colors is reduced in each round.
 We show that the element $s_j$ can be found using binary search, without knowing the neighbors sets $S(\phi(u))$, but only knowing the number of intersections with neighbors' sets in a specific range. When solving 2-distance-coloring using this idea in a straightforward way, each node needs to receive all the subsets from 2-distant neighbors. This causes all nodes to send messages with size of at least $O(\Delta)$ for each node in order to receive messages with all of its distance-$2$ neighbors colors. (Each neighbor of a given vertex sends it the colors of all its own neighbors.) Using this information a vertex can compute the available color to choose. However, this approach exchanges much more information than needed and can be optimized for restricted bandwidth models. We provide this optimization in Section \ref{LinialHighLevelDescription}.

\subsection{High level description}\label{LinialHighLevelDescription}

Our technique does not use an ordinary set, but an ordered set, thus we can perform a binary search. 
The goal of the binary search, for a vertex $v \in V$, is finding an element in $S(v)$, that does not belong to any set $S(u)$, for  $u$ in the 2-hop-neighborhood of $v$. To this end, for each vertex we define a range, that initially contains all elements of $S(v)$. Then we reduce the range size by a factor of at least 2 in each stage of the binary search. Eventually, each $v \in V$ reduces its range to contain a single element that does not belong to any range in its 2-hop-neighborhood. However, there is a cost in running time because each binary search requires $O(\log k)$ rounds where $k$ is the size of the colors palette of the current stage. 

The technique in high level is that each vertex knows its 1-distance neighbors subsets that are based on the coloring $\phi(u)$ and marked $S(\phi(u))$ and computes the number of intersections with these subsets.
Each node holds two indices that constitute the beginning and end of the relevant range. The neighbors are aware of those indices and refer to the beginning index by {\em left} and the ending index by {\em right}. On every iteration the number of intersecting values of the subsets is reduced by a factor of at least two, simply by counting the elements in each half of a range and choosing the half-range with less intersections.
On each round each vertex receives the number of intersecting values for both their left half and right half from its neighbours and decides whether in the next round it will use the left half of the range or the right half, and update its $left,right$ indices accordingly. The selected half of the range is the one with fewer intersecting values.

Next, we provide the pseudocode of our algorithm (see Algorithm 1 below), called {\em 2-Distance-Linial}, and analyze its correctness and running time.



\begin{algorithm}
	\begin{algorithmic}[1]
    \STATE Let $t$ be the size of each set in the current phase /* All sets have the same size in the same phase */
    
    \STATE Let $S(v) = S(\phi(v)) = \{ s_1(v),s_2(v),...,s_t(v) \}$ be the ordered set produced by the algorithm of Linial for the vertex $v$. /* $S(\phi(v))$ is computed locally by Linial's algorithm as a function of $\phi(v)$. */
    
    \STATE Let $S(u_i) = S(\phi(u_i)) = \{ s_1(u_i),s_2(u_i),...,s_t(u_i) \}$
     
    \STATE /* From Linial's proof we know that every $S$ contains an element $s_i$, such that $s_i$ belongs to $S$ and does not belong to any other $S(\phi(u_j))$ with $\phi(u_j) \neq \phi(u_i)$. */
    
    \STATE The vertex v performs a binary search in the set  S, with assistance of its neighbors, as follows:
    
    \STATE $left = 1, right =t$

        \WHILE{$left \neq right$} 
        
        \FOR {all $u,w \in \Gamma^+(v)$ in parallel}
        
        \STATE $Int_l(w ,u)$ = Number of intersections of $S(w)$ and $\{ s_i(u) \ | \ i \in [left,...,\frac{right}{2}  -1] \}$
        
        \STATE $Int_r(w ,u)$ = Number of intersections of $S(w)$ and $\{ s_i(u) \ | \ i \in [\frac{right}{2}  , right] \}$
        
        \ENDFOR

        \STATE 
        The node $v$ sends in parallel to all $u \in \Gamma(v)$ the sums: $\sum_{w \in \Gamma^+(v)} Int_l(w ,u)$ and $\sum_{w \in \Gamma^+(v)} Int_r(w ,u)$ 

        \STATE The node $v$ receives from all its neighbors $u$ the messages $sum_l(u ,v)$ and $sum_r(u ,v)$.

        \STATE The vertex $v$ computes the sums  $sum_l = \sum_{u \in \Gamma(v)} sum_l(u,v)$ and $sum_r = \sum_{u \in \Gamma(v)} sum_r(u,v)$. 
            \IF {$sum_r \geq sum_l$}
            
                \STATE $right = \frac{right}{2} -1$
                
                \STATE Send to all neighbors "left chosen"
        
            \ELSE 
            
                \STATE $left = \frac{right}{2}$
                
                \STATE Send to all neighbors "right chosen"
                
            \ENDIF
            
            \STATE The new set $S(v)$ for the next phase of the binary search is $S(v) = \{ s_{left}(v),...,s_{right}(v) \}$
            
            \STATE Receive all neighbors [left $\backslash$  right] choices and compute $S(u)$ for all $u \in \Gamma(v)$
        
        \ENDWHILE

        \STATE return the color $s_{left}(v)$  \ \ \ \ \ \ /* now $s_{left}(v)= s_{right}(v)$ */
        
        \end{algorithmic}
    \caption{2-Distance-Linial's algorithm phase} 
\end{algorithm}

\subsection{Proof and run time analysis}\label{LinialAnalisys}
\begin{lem}\label{linialEnoughColors}
	After each invocation of 2-Distance-Linial's algorithm phase the coloring $\phi$ remains proper.
\end{lem}
\begin{proof}
	From the original proof of Linial's algorithm, it follows that a subset $S(v)$ of possible colors generated according to Linial's original algorithm for each vertex necessarily contains at least one color $s_i \in S(v)$ that does not belong to any subset $S(u)$ of the neighbors $u$ at distance at most 2 from $v$. Moreover, the number of intersections (i.e., identical elements) between a vertex $v$ and all its distance-$2$ neighbors is smaller than $|S(v)|$. Denote the set of the smallest $ \lfloor \frac{1}{2}|S(v)| \rfloor$ elements of $S(v)$ by $S(v, Left)$ and the set of the other (larger) elements by $S(v, Right)$. It holds that if $sum_l \ge \lfloor \frac{|S(v)|}{2} \rfloor$ then $sum_r < \lceil \frac{|S(v)|}{2} \rceil$. Thus, at least one set ($S(v, Left)$ or $S(v, right)$) has less than half the conflicts of $v$ with its $2$-hop-neighbors. Moreover, the number of conflicts with that set is smaller than its size.
 
	If we continue with this argument recursively $\log |S(v)|$ times, then each time $(right - left)$ is reduced by a mutiplicative factor of $2$. In addition, the number of conflicts is reduced at least by a multiplicative factor of $2$ as well. Thus in each phase, either $S(v, Left)$ or $S(v, Left)$ contains an element with no conflicts in the $2$-hop-neighborhood. Such a subset is selected for the next phase. Therefore, after $\log{|S(v)|}$ phases, for each $v \in V$, an index is found with $right = left$ and $sum_l = sum_r = 0$, and there are no conflicts with the $2$-hop-neighborhood of $v$ at that index. 
\end{proof}

\begin{lem}\label{LinialRunningTime}
	The running time of 2-distance-Linial in the CONGEST model is $O(\log^{*}n \cdot \log \Delta + \log{\log{n}})$.
\end{lem}
\begin{proof}
	First we compute the running time of a single iteration (lines 7-24). Denote the number of colors produced in the current phase of Linial's algorithm by $O(t^2)$, and the size of the set produced in line 2 by $t$. Consequently, each message sent in line 12 requires $2 \cdot \log t$ bits, i.e., $\log t$ bits for sending the sum for the left half, and $\log t$ bits for sending the sum for the right half. Thus, such messages can be sent in one round in the CONGEST model. It takes $O(\log t)$ rounds for $left, right$ to become equal to one another. Thus, the number of a phase if $O(\log t)$.
	Next, we compute the running time of the entire algorithm, using the ranges of colors computed by the algorithm. The number of colors in the first phase is $O(\Delta^4 \cdot \log^2 n)$, in the second phase it is $O(\Delta^4 \cdot \log^2 \log n)$ ,..., after $\log^*n$ phases it is $O(\Delta^6)$, and in the final phase it becomes an $O(\Delta^4)$-proper-coloring. So the first phase running time is $O(\log ({\Delta^4 \cdot \log^2 n})) = O(\log {\Delta} + \log \log n)$, the second phase running time is $O(\log ({\Delta^4 \cdot \log^2 \log n})) = O(\log {\Delta + \log^{(3)} n})$, etc. After $\log^*n$ phases the overall running time is $O(\log^*n \cdot \log \Delta + \log \log n)$. Note that  $\sum_{i = 2}^{\log^* n}{\log^{(i)}n} = O(\log \log n)$.
\end{proof}

Next, we provide an improvement, which removes the $O(\log \log n)$ factor from Lemma \ref{LinialRunningTime}.
To this end, we perform each binary search for $O(\log \Delta)$ phases, rather than $O(\log t)$. Moreover, in each phase we send just one bit to indicate whether the left half is chosen or the right one, rather than sending indices of ranges. This information is sufficient to compute the new range from the previous one. After $O(\log \Delta)$ phases, we obtain a consecutive range of size $O(t/\Delta^2) = O(\log n)$. Recall that the initial range size is $t = O(\Delta^2 \cdot \log n)$, which is a square root of the number of colors. After $O(\log \Delta)$ phases, the new range defines a bit-string of size $O(\log n)$ that represents whether there is a conflict for each element in $\{s_{left}(v),...,s_{right}(v)\}$. This string is then sent directly to $v$ by each of its $1$-hop-neighbors. This is done using $O(\log n)$-bits messages, within one round in parallel by all neighbors. Then $v$ finds an index $i \in O(\log n)$ of a bit with a $0$ value, which exists since there is an index without conflicts in the range $\{s_{left},...,s_{right}\}$. The resulting color with no conflicts is $s_{left + i}(v)$. We summarize this in the next theorem.
\begin{thm} \label{resultalglin}
A proper distance-$2$ coloring with $O(\Delta^4)$-colors can be computed in $O(\log \Delta \cdot \log^* n)$ rounds in the CONGEST model.
\end{thm}

Theorem \ref{resultalglin} demonstrates that for each of the $\log^* n$ iterations of Linial's algorithm, $O(\log \Delta)$ rounds are performed to compute the color for the next iteration. We now argue that it is sufficient to perform just two iterations with $O(\log \Delta)$ rounds, while the remaining $O(\log^* n)$ iterations require $O(1)$ rounds each. The idea is similar to an improvement from $O(\Delta \log^* n)$ to $O(\Delta + \log^* n)$ of \cite{HKM2020}. Specifically, after $2$ iterations the number of colors becomes $O(\Delta^4  \log^2 \log n)$. If $\Delta^4 > \log^2 \log n$, then this is an $O(\Delta^8)$-coloring. 
It can be converted into an $O(\Delta^4)$-coloring within a single iteration, using a field of size $\Theta(\Delta^2)$, and polynomials of degree 4. Indeed, each of the current $O(\Delta^8)$ colors can be assigned a unique polynomial, and each such polynomial has at least one non-intersecting point with any $\Delta^2$ others. 
These are used for computing new colors in a range of size $O(\Delta^4)$. The other possibility is that $\Delta^4 \leq \log^2 \log n$. Then, instead of performing a binary search, one can directly send $\Delta$ messages, each of which consists of a color in a range of $O(\Delta^4  \log^2 \log n) = O(\log^4 \log n)$.
Since $\Delta \leq \sqrt{\log \log n}$, the number of bits in the concatenation of these $\Delta$ messages is $\mbox{poly} (\log \log n)$, so it can be sent over an edge within $O(1)$ rounds of the CONGEST model. 

\begin{col} \label{colresultalglin}
A proper distance-$2$ coloring with $O(\Delta^4)$-colors can be computed in $O(\log \Delta + \log^* n)$ rounds in the CONGEST model.
\end{col}

Note: when counting the number of intersection in a 2-hop-neighborhood of a vertex $v \in V$, some intersections may be counted more than once. For example, suppose that $v$ has two neighbors $u,w$, and those vertices have a common neighbor $z$ at distance $2$ from $v$. Then the number of intersections of $v$ and $z$ would be counted at least twice. (Because each of $u$,$w$ counts these intersections.) Nevertheless, the overall number of counted intersections of $v$ is bounded by the product of the number of edges adjacent to neighbors of $v$ and the maximum number of intersections of $v$ with another vertex.


\noindent This is smaller than $S(v)$, because the number of edges adjacent to $v$'s neighbors is at most $\Delta^2$. We elaborate on the issue of overcounting in more general scenarios in Section 4.3.


\section{Distance-$2$ coloring with $O(\Delta^2)$ colors in $O(\Delta  \cdot \log{\Delta} + \log^* n)$ rounds} \label{secfastcol}

\subsection{High level description}

In this section we provide an improved algorithm for distance-$2$ coloring of $G$ using $O(\Delta^2)$ colors, which is distance-1 coloring of $G^2$. The improvement is from $O(\Delta^2 + \log^* n)$ rounds to $O(\Delta \cdot \log  \Delta + \log^* n)$ rounds in the CONGEST model. 
(A procedure that requires $O(\Delta^2 + \log^* n)$ rounds is provided in Appendix A, for completeness.)
The improved result is achieved by applying our technique to the algorithm of \cite{BEG18}  that provides an $O(\Delta)$ coloring of an input graph $G$ in $O(\sqrt \Delta \log {\Delta} + \log^{*} n)$ rounds. 
This gives rise to $O(\Delta^2)$-coloring of $G^2$ in $O(\Delta \cdot \log  \Delta + \log^* n)$ rounds.

The algorithm of \cite{BEG18} is based on the following notions.\\
A {\em $p$-defective coloring} is a vertex coloring such that each vertex may have up to $p$ neighbours with its color.\\
An {\em $p$-arbdefective coloring} is a vertex coloring, such that each subgraph induced by a color class of the coloring has arboricity bounded by $p$. The {\em arboricity} is the minimum number of forests into which the edge set of a graph can be decomposed.

The algorithm of \cite{BEG18} for distance-$1$ coloring consists of three stages:\\
1. Computing $O(\sqrt{\Delta})$-defective $O(\Delta)$-coloring of $G$ in $O(\log^{*} n)$ time.
\newline
2. Computing $O(\sqrt{\Delta})$-arbdefective $O(\sqrt{ \Delta})$-coloring of $G$ in $O(\sqrt \Delta)$ time.
\newline
3. Iterating over the $O(\sqrt{\Delta})$ color classes of step 2, and computing a proper coloring of $G$ iteratively. Each iteration of Stage 3 requires a constant number of rounds.

As a first step of our extension of this scheme to work in $G^2$, we introduce the following {\em proxy} communication method. The goal is establishing a single path between any pair of vertices that need to communicate, and are at distance at most 2 one from another. (In section \ref{speedup}, we generalize this to vertices of distance $k$ one from another.) This way, the desired information is passed only once, which improves communication costs, and avoids miscalculations caused by duplicated data. To this end, for each vertex, a BFS tree of radius two that is rooted at the vertex is computed. This computation starts with sending the list of neighbors of each vertex to all its neighbors. This requires $O(\Delta)$ rounds, because in each round each vertex sends the information to a neighbor that has not been sent yet.  Then, each vertex is aware of the neighbors of its neighbors. For any distance-2 neighbor, it knows the immediate neighbors that connect to it, and selects exactly one of these neighbors. The selected neighbor is referred to as {\em proxy}. Note that a vertex $v$ with its proxy nodes and their neighbors form a BFS of radius two that contains the two-hop neighborhood of $v$.


In Sections \ref{runningArbLinial} - \ref{arbLinial}, we describe the generalizations and modifications in the above-mentioned stages (1) - (3), for coloring $G^2$ in $O(\Delta \cdot \log \Delta +\log^{*} n)$ rounds in the CONGEST setting. For these computations we need the above-mentioned notion of proxy nodes.

The running time analysis assumes that the proxy nodes have already been computed. Otherwise, an $O(\Delta)$ term should be added. However, this does not affect the overall running time of the entire algorithm of Section 3, which is $O(\Delta \log \Delta + \log^* n)$.

\subsection{Detailed description of the algorithm}

\subsubsection{Our variant for distance-$2$ defective coloring}\label{runningArbLinial}
The computation of $O(\Delta)$-defective $O(\Delta^2)$-coloring of $G^2$ proceeds as follows. \newline
The procedure starts by computing a proper $O(\Delta^4)$-coloring of $G^2$, using Algorithm 1. Next, find a prime $q = \Theta(\Delta)$, such that the number of colors is bounded by $q^4$. Note that each color is represented by a tuple $\langle a,b,c,d \rangle$, where $a,b,c,d \in  Z_q$, and $Z_q$ is a field modulo $q$. Assign each color a unique polynomial $p(x) = a + b \cdot x + c \cdot x^2 + d \cdot x^3$. Assign each vertex a polynomial according to its color. We say that two polynomials $p(x),p'(x)$ {\em intersect} at the value $t$,$0 \leq t \leq q - 1$, if $p(t) = p'(t)$. Next, each vertex $v \in V$ finds a value $t, 0 \leq t \leq q-1$, such that $p(t)$ intersects with the minimum number of polynomials of vertices of distance at most 2 from $v$. This is done as follows, by a binary search. The vertex $v$ sends its polynomial to its $1$-hop neighbors. Each of these neighbors $u \in \Gamma(v)$ computes the number of intersections of $v$'s polynomial with polynomials of neighbors $w$ of $u$, such that $u$ is the proxy for $\{v,w\}$. In addition, each  $u \in \Gamma(v)$ computes the number of intersections of its polynomial with that of $v$. The number of intersections is computed for each half of the range $\{0,1,....,q - 1\}$, i.e., $\{0,1,....\left \lceil q/2 \right \rceil \}$ and $\{\left \lceil q/2 \right \rceil + 1, \left \lceil q/2 \right \rceil + 2,.... q-1]\}$. This information is returned to $v$ by all its $1$-hop neighbors. Then $v$ knows how many intersections with its $2$-hop neighbors its polynomial has in $\{0,1,....\left \lceil q/2 \right \rceil \}$ and $\{\left \lceil q/2 \right \rceil + 1, \left \lceil q/2 \right \rceil + 2,.... q-1]\}$. The half-range with fewer intersections is selected for the next iteration of the binary search. This is repeated for $\log q$ iterations, until the range contains a single element $t \in \{0,1,...,q - 1\}$. The color of $v$ returned by the procedure is $\langle t, p(t) \rangle$. This completes the description of the procedure. Its correctness and running time are analyzed below.

\begin{lem}
The procedure computes an $O(\Delta)$-defective $O(\Delta^2)$-coloring of $G^2$.
\end{lem}
\begin{proof}
Every pair of vertices $u,v$ at distance at most $2$ one from another obtain distinct colors, according to the $O(\Delta^4)$-coloring of  $G^2$. Thus, $u,v$ are assigned distinct polynomials $p(u) \neq p(v)$. The polynomials are of degree 3. Therefore, each vertex has at most $3$ intersections of its polynomial with each of its 2-hop neighbours' polynomials. Since $v$ has at most $\Delta^2$ neighbors in its $2$-hop-neighborhood, each vertex has $O(\Delta^2)$ possible intersections in its polynomial. Each node has $q = O(\Delta)$ possibilities for selecting $\langle k, p(k) \rangle$, since $k \in \{0,1,2...,q-1\}$. Thus, by the pigeonhole principle, there must exist an element  $\langle k, p(k) \rangle$ with at most $O(\Delta)$ intersections of the polynomial of $v$ with all its neighbors' polynomials.

Next, we show that the number of intersections with polynomials in $2$-hop-neighborhoods received by each vertex is correct.
Assume that a vertex $v$ receives messages from its 1-hop neighbors $u_1,u_2,...,u_{\Delta}$ with the values $k_1,..,k_{\Delta}$, respectively. It holds that each vertex at distance $2$ from $v$ has exactly one proxy node in $\{u_1,u_2,...,u_{\Delta}\}$.
Thus each intersection of the polynomial of $v$ with a polynomials of a vertex in $v$'s $2$-hop-neighborhood is counted exactly once, by the proxy node of that vertex.
Assume that when a vertex $v$ chooses a half, without loss of generality, it is the left half.
In this case, according to our algorithm, that left half has less conflicts than the right half. This means that the number of intersections in the half-range is smaller by at least a factor of 2 than the number in both halves of the range.  Therefore, once the range size is reduced by a factor of $q$, the number of intersections must also reduce by a factor of at least $q$. Since $q = O(\Delta)$ and we have already shown that the initial number of intersections is $O(\Delta^2)$,  when the range size becomes 1, there are at most $O(\Delta)$ intersections. Therefore, the resulting coloring is $O(\Delta)$-defective $O(\Delta^2)$-coloring.
\end{proof}

\begin{lem}
The running time of the procedure is $O(\log \Delta + \log^* n)$.
\end{lem}
\begin{proof}
First, in order to compute $O(\Delta^4)$ coloring we employ Algorithm 1 that has running time $O(\log \Delta + \log^* n)$, by Corollary \ref{colresultalglin}. 
The remaining part of the  procedure is a  binary search on a range of size $O(\Delta^4)$, and thus requires $O(\log (\Delta^4)) = O(\log \Delta)$ phases, each of which consists of a constant number of rounds. The overall running time is $O(\log \Delta + \log^* n)$.
\end{proof}

\subsubsection{Algorithm for Distance-$2$ Arbdefective Coloring}\label{AGArbdefective}

For a graph $G$, given an $O(\Delta)$-defective $O(\Delta^2)$-coloring of $G^2$, we would like to produce an $O(\Delta)$-arbdefective $O(\Delta)$-coloring of $G^2$  within $O(\Delta)$ rounds.
The algorithm is as follows: 

This algorithm extends the ideas of \cite{BEG18} to work in $G^2$ in the CONGEST model. In that paper the authors devised an $O(\sqrt{\Delta})$-arbdefective $O(\sqrt{\Delta})$-coloring algorithm for $G$ with $O(\sqrt{\Delta} + \log^* n)$ rounds. The main idea of the algorithm is as follows. In each round, each vertex counts how many of its neighbors have the same color as its own. The number of such neighbors is the {\em number of conflicts}. If a vertex has too many conflicts, it selects a new color, using a certain function. Otherwise, the vertex finalizes its color. The original algorithm (see Section 6 in \cite{BEG18}) proceeds for $O(\sqrt{\Delta})$ rounds, and selects the round with the smallest number of conflicts. By the pigeonhole principle, there must be a round in which the number of color conflicts is at most $O(\sqrt{\Delta})$. However, 
computing the number of conflicts with all $2$-hop-neighbors is expensive, since each original round can take up to $O(\Delta)$ rounds, when applied in a straightforward way to $G^2$. 
To improve this, each vertex collects information about conflicts from its $2$-hop-neighborhood in a bit-efficient manner. Specifically, a vertex receives from each of its 1-hop neighbors the number of conflicts it has with 2-hop-neighbors, instead of lists of their colors. 
During the execution of the algorithm, in each iteration the conflicts are counted, and if the total number of 2-distance conflicts is below a predefined $t$, the vertex finalizes the current color.
For any $t \in [1, \Delta^2]$, this stage requires $O(\frac{\Delta^2}{t})$ time and it results in $O(t)$-arbdefective $O(\frac{\Delta^2}{t})$-coloring of $G^2$. Setting $t = O(\Delta)$ results in $O(\Delta)$-arbdefective $O(\Delta)$-coloring in time $O(\Delta)$.  See pseudocode of Algorithms 2 - 3 below. (Each color in an initial $O(\Delta^2)$-coloring is represented by an ordered pair $\langle a, b \rangle$, where $a,b \in O(\Delta)$. When Algorithm 3 terminates, the resulting color resides in the $b$-coordinate, and it is in the range $[0,1,...,O(\Delta)]$.) Next, we analyze Algorithm 3.


\begin{algorithm}
	\begin{algorithmic}[1]

    \STATE /* This procedure is performed internally by a vertex, within 0 rounds */
    \STATE numberOfConflicts = 0
    \FOR{$i = 1,2,...,n$}
	    \IF {$b_0$ = $b_i$}
	        \STATE numberOfConflicts = numberOfConflicts + 1
	    \ENDIF
	\ENDFOR
	\STATE return numberOfConflicts
    \end{algorithmic}
    \caption{Procedure number-of-conflicts  ($\langle a_0, b_0 \rangle$, { $\langle a_1, b_1 \rangle$) .. $\langle a_n, b_n \rangle$ )}}
\end{algorithm}
\begin{algorithm} 
	\begin{algorithmic}[1]
		
		\STATE We are given a $p$-coloring $\phi$. Denote $q$ as the smallest prime number such that $q \geq \sqrt{p}$. The parameter $maxDefect$ is the maximal arb-defect allowed for coloring.

		\STATE Denote $\phi(v) = \langle a , b \rangle$,  where $a, b \leq q$
		\WHILE {$\phi(v) \neq \langle 0, b \rangle$}
		
		\STATE Denote by $conflicts(v, u) \leftarrow$  number-of-conflicts($\phi(v)\bigcup_{u_i \in \Gamma(u)}\phi(u_i)$) 
		
		\FOR{$i = 1,2,..,deg(v)$ in parallel}  
  
            \STATE send to the $i$th neighbor of $v$, which is $u_i$, the message $conflicts(u_i, v)$ 

        \ENDFOR
  
		\STATE Receive all $conflicts(v, u_i)$  messages from neighbors 
		
		\IF {$\sum_{(v, u_i) \in E}{conflicts(v, u_i)} \leq maxDefect$}
		
			\STATE $\phi(v) = \langle 0, b \rangle$
			\STATE Send "Done" to all neighbors
		\ELSE
			\STATE $\phi(v) = \langle a, a + b \mbox{ mod } q\rangle$
			\STATE Send "Not done" to all neighbors
		\ENDIF
		
		\STATE Receive all "Done", "Not Done" messages from neighbors, and compute $\phi(u_i)$ for $i = 1,2,...,deg(v)$ 
		\ENDWHILE
	\end{algorithmic}
	\caption{2-Distance AG Arbdefective Coloring($maxDefect$ = t)} \label{alg2distarb}
\end{algorithm}


\begin{lem} \label{maxd}
After running 2-Distance AG Arbdefective Coloring for $\lceil 2\Delta^2/maxDefect \rceil$ rounds, 
all vertices have colors of the form $\langle 0, b \rangle$, and each color class has arboricity at most $maxDefect$ in $G^2$. 
\end{lem}
\begin{proof}
According to \cite{BEG18}, each vertex can have at most two conflicts with each of its $2$-hop-neighbors during execution. Thus, by pigeonhole principle, for each $v \in V$ there is a round $i \in \{1,2,...,\lceil 2\Delta^2/maxDefect \rceil \}$, such that $v$ has at most $maxDefect$ conflicts in its $2$-hop-neighborhood. In such round the color of $v$ becomes of the form $\langle 0, b \rangle$. For the sake of analysis, orient each edge $(v,u) \in E(G^2)$ towards the endpoint that changed its color to the form $\langle 0, b \rangle$ earlier. (Ties are broken using IDs, towards smaller IDs.) Consider the endpoint that obtained the form $\langle 0, b \rangle$ after the other endpoint, or in the same time. The number of conflicts at that time was at most $maxDefect$, hence this endpoint has at most $maxDefect$ outgoing edges to vertices with its $b$-value. Consequently, in the end of the execution, in each subgraph of $G^2$ induced by a color $i \in \{1,2,...,\lceil 2\Delta^2/maxDefect \rceil\}$ , each vertex has up to $maxDefect$ outgoing edges. Thus the arboricity of each such subgraph is at most $maxDefect$. (Assign each outgoing edge of a vertex a distinct label from $\{1,2,...,maxDefect\}$. When this is done for all vertices, the edges of the same label form a forest.)
\end{proof}

\subsubsection{Iterative Algorithm for Distance-$2$ Proper Coloring}\label{arbLinial}
In this subsection we describe an algorithm that produces an $O(\Delta^2)$-proper-coloring of $G^2$ within $O(\Delta \cdot \log \Delta + \log^* n)$ rounds. This algorithm is based on a the technique devised in \cite{BEG18}, but our algorithm extends this technique to work for $G^2$. For more details regarding this technique, see section 3 in \cite{B15}.
Our new algorithm starts with computing an $O(\Delta)$-arbdefective $O(\Delta)$-coloring $\varphi$ for $G^2$.
The coloring $\varphi$ constitutes a partition of the graph into $O(\Delta)$ color sets $V_1,V_2,...,V_d$, $d \in O(\Delta)$. Each color class is $O(\Delta)$-arbdefective. This means that the arboricity of a subgraph of $G^2$ induced by $V_j$, $j \in O(\Delta)$, is bounded by $O(\Delta)$. Moreover, each pair of vertices $u,v$ at distance at most 2 one from another have a parent-child relation in a certain forest. Specifically, when an arbdefective coloring is computed with Algorithm \ref{alg2distarb}, if $u,v$ terminate (arrive to step 10 of the algorithm) in distinct rounds, then the parent is the vertex that terminated earlier. Otherwise, the parent is the vertex with lower ID. Vertices do not have to know their parents explicitly.

The algorithm iterates over $i = 1,2,...,d$. In each iteration $i$ the algorithm computes a new color $\varphi'$ for all of the nodes with color $\phi(v)=i$, using at most $O(\Delta^2)$ colors. To this end, each vertex constructs a set of polynomials, and finds a polynomial $P$ in this set, such that:\\
(1) The number of intersections of the polynomial $P$ with colors $\varphi'$ of vertices in its 2-hop neighborhood that already selected such colors in previous rounds is as small as possible. \\
(2) The number of intersections of $P$ with polynomials of its parents in its $2$-hop-neighborhood that are active in the same round $i$ is as small as possible.

The construction of the polynomial set of a vertex $u \in V_i$ is performed as follows. 
Let $q = O(\Delta)$ be a prime, such that $q > c \cdot \Delta$, for a sufficiently large constant $c \geq 1$.
We represent the color $\varphi(u)$ by $\langle a,b \rangle$, where $0 \leq a,b < q$.
The set of polynomials of $u$ is $\{ a \cdot x^2 + b \cdot x + j \ | \ j = 0,1,...,q-1 \}$.
The number of polynomials in the set is $q = O(\Delta)$.

According to (1), our goal is finding a polynomial $P$ in the set of $u$, such that the number of vertices at distance at most $2$ from $u$ with the following property is minimized.

(*) For a vertex $w$ that already has a color $\varphi'(w) = \langle a_w',b_w' \rangle$, there exists $t \in {0,1,...,q-1}$, such that $\langle  a_w',b_w' \rangle = \langle t, P(t) \rangle$. 

According to the Pigeonhole principle, there must be a polynomial in the set of size $q > c \cdot \Delta$, for which at most $q/2$ vertices at distance at most $2$ satisfy this property. This is because each vertex satisfies this property for at most one polynomial in this set (the set consists of non-intersecting polynomials), and the number of vertices at distance at most $2$ is at most $q^2/2$. Our goal is finding such a polynomial.
The challenge is that when running a naive version of this algorithm in the CONGEST model, every vertex needs to know its 2-hop neighbours' colors. Sending this information requires $O(\Delta)$ rounds. Next we describe an optimization that requires only to compute how many intersections there are in sets of polynomials. This speeds up the running time from $O(\Delta)$ to $O(\log{\Delta})$.


Next, we describe how each vertex selects the desired polynomial $a \cdot x^2 + b \cdot x + j$ from its set within $O(\log \Delta)$ rounds. This is done using a binary search on $j$. To this end, each vertex has to inform its neighbors about its set of polynomials. Even though there are $q$ polynomials in the set, this is done just within one round, as follows. Given a set of polynomials  $\{ a \cdot x^2 + b \cdot x + j \ | \ j = 0,1,...,q-1 \}$ of a vertex $u$, only the coefficients $a,b$ are sent to the neighbors of $u$. (Each coefficient requires $O(\log \Delta)$ bits.) Then the neighbors can reconstruct the set of polynomials from $a,b$, since they know that $j$ runs from $0$ to $q - 1$.
Next, every vertex initialize $start=1$ and $end = q$ and defines two ranges. The ranges are $low = [start, \lceil \frac{end - start}{2} \rceil]$ and $high = [\lceil \frac{end - start}{2} \rceil + 1, end]$. At the first step each vertex sends to each of its neighbors $w$ the number of intersections of colors $\varphi'$ with polynomials that have $j$ in range $low$, as well as the number of intersections for $j$ in range $high$. 
In the next step each vertex receives from its neighbors the number of such intersections in ranges $low$ and $high$ in its 2-hop-neighborhood.
Then each vertex decides for its new $start$ and $end$ according to the half range in which there are fewer intersection with its polynomials. 
Consequently, after halving $O(\log{\Delta})$ times, the range contains just a single value $\hat{j}$. It defines a single polynomial from the set, which is $a\cdot x^2 + b \cdot x + \hat{j}$.

The next lemma provides a helpful property of the polynomials, which will assist us to compute the coloring of a set $V_i$, given colorings of $V_1,V_2,...,V_{i-1}$. Specifically, each vertex in $V_i$ can select a polynomial with sufficiently many coordinates that do not intersect with colors that heave already been chosen by its $2$-distance neighbors in $V_1,V_2,...,V_{i-1}$. Then, out of these coordinates, a coordinate that does not intersect with any selected polynomial in 2-distance in $V_i$ can be found. This is a generalization of ideas that appear in Section 3 of \cite{B15} and analyzed in Section 3.4 of \cite{B15}. (The results in \cite{B15} are for 1-distance coloring, while here we analyze 2-distance coloring.)

\begin{lem} \label{lem:iter}
Suppose that we are given a graph with $O(\Delta)$-arbdefective $O(\Delta)$-coloring $\varphi$ of $G^2$, that partitions the input graph into subsets $V_1,V_2,...,V_{O(\Delta)}$, according to color classes of $\varphi$. Moreover for an integer $i \geq 0$, suppose that we already have a proper 2-distance coloring $\varphi'$ for $V_1,V_2,...,V_{i-1}$. Then we can find a polynomial $P = a \cdot x^2 + b \cdot x + \hat{j}$ for each vertex $u$ with $\varphi(u) = i$, such that at least half of the elements in the set  $\{ \langle 0, P(0) \rangle, \langle 1, P(1) \rangle..., \langle q-1, P(q-1) \rangle \}$ do not appear as $\varphi'$ colors in the 2-hop neighborhood of $u$.
\end{lem}
\begin{proof}
This is shown by the pigeonhole principle. For a color class $i$ and a vertex $u$ with $\varphi(u) = i$, the number of 2-hop neighbors of $u$ that already computed the color $\varphi'$ is bounded by the size of the 2-hop-neighborhood, which is $\Delta^2$. 
Our algorithm finds a polynomial $P$ in the set of $u$ with the smallest number of elements $\langle t, P(t) \rangle$ that appear as $\varphi$ colors in the 2-hop-neighborhood of $u$. Moreover, for each pair of polynomials $P_1,P_2$ of $u$, and any $t',t''$, it holds that $\langle t', P_1(t') \rangle \neq \langle t'', P_2(t'') \rangle$ because all polynomials of $u$ do not intersect one with another. Hence the polynomial $P$ has at most $\Delta^2/q < \Delta$ elements as $\varphi'$ colors in the 2-hop-neighborhood of $u$.  On the other hand, $q > c \cdot \Delta$, for a sufficiently large constant $c$, thus at least half of the elements in the set  $\{ \langle 0, P(0) \rangle, \langle 1, P(1) \rangle..., \langle q-1, P(q-1) \rangle \}$ do not appear as $\varphi'$ colors in the 2-hop neighborhood of $u$,
\end{proof}

According to Lemma \ref{lem:iter}, it is possible to iterate over the color classes of the arbdefective coloring $\varphi$, for $i=1,2,...,O(\Delta)$.
In each iteration $i$, each vertex in the color class $i$ obtains a single polynomial with the properties stated in the lemma. Specifically, it has sufficiently many elements that still can be used for their $\varphi'$ color. Specifically, the number of elements is larger (by a factor greater than 2) than the number of their parents in $G^2$. Consequently, a variant of Linial's algorithm that considers only parents in the $2$-hop-neighborhood can be executed. (For the case of distance-1 coloring, this is a well-known extension of Linial's algorithm, which is called arb-Linial \cite{BE08}.) In the case of distance-2 coloring it can be computed in $O(\log \Delta)$ phases in the same way as in Corollary \ref{colresultalglin}, but considering only 2-hop-parents, rather than entire 2-hop-neighborhood. (Recall that a vertex can deduce the parent-child relationship of a pair of its neighbors, by inspecting their termination round in Algorithm 3 and their IDs.) The result of this computation is a single coordinate of the polynomial, which gives rise to a color $\varphi'$ for the vertex, which is distinct from all $\varphi'$-colors in its $2$-hop-neighborhood. 
Hence, we obtain the following Corollary.

\begin{col} \label{square-col}
It is possible to compute a proper distance-$2$ coloring with $O(\Delta^2)$ colors within
$O(\Delta \cdot \log \Delta + \log^* n)$ rounds in the CONGEST model.

\end{col}

\subsubsection{Coloring $G^2$ using $(\Delta^2 + 1)$ colors in $O(\Delta^{\frac{3}{2}} \cdot \log \Delta + \log^* n)$ rounds}
In this section we show how to reduce the number of colors to $(\Delta^2 + 1)$. To this end, we parameterize the steps of our scheme in a different way, as follows. \newline
1. Compute an $0$-defective $O(\Delta^2)$-coloring of $G^2$, i.e., a proper $O(\Delta^2)$-coloring.
\newline
2. Compute $O(\sqrt{\Delta})$-arbdefective $O(\Delta^{3/2})$-coloring of $G^2$.
\newline
3. Iterate over the $O(\Delta^{3/2})$ color classes that were generated in step 2, and compute a proper coloring of $G^2$ iteratively, using $\Delta^2 + O(\Delta^{3/2})$ colors.
\newline
4. Apply a simple reduction to produce a $(\Delta^2 + 1)$coloring of $G^2$ from a $\Delta^2 + O(\Delta^{3/2})$-coloring.

The steps are performed as follows. Step 1 is obtained by applying Corollary \ref{square-col}. This step requires $O(\Delta \cdot \log \Delta + \log^* n)$ rounds. Step 2 is obtained by applying Lemma \ref{maxd} with $maxDefect = \sqrt \Delta$ and  $q = \Theta(\Delta^{3/2})$. This step requires $O(\Delta^{3/2})$ rounds. Step 3 is performed similarly to Section \ref{arbLinial}, but now we have $O(\Delta^{3/2})$ color classes to iterate on, rather than $O(\Delta)$. On the other hand, the arboricity of each of them is significantly smaller, and consequently a proper coloring from a range of size $\Delta^2 + O(\Delta^{3/2})$ can be computed. To this end, let $q$ be a prime, such that $q > \Delta + O(\sqrt{\Delta})$, $q < 2 \cdot \Delta + O(\sqrt{\Delta})$. Then each vertex $v$ in $G^2$ has at most $\Delta^2$ neighbors which have finalized their colors in previous rounds, thus it cannot select their colors. By the pigeonhole principle, there exists a polynomial 
$P$ of $v$, such that $P$ intersects with at most $\frac{\Delta^2}{q}$ neighboring colors. 
This polynomial is defined over a field of size $q$, and there is a value $t \leq \Delta^2/q + O(\sqrt{\Delta})$, such that $\langle t,P(t) \rangle$ does not intersect with any neighboring polynomial, nor with any neighboring color. Thus, $\langle t,P(t) \rangle$ is selected as the color of $v$, and it is from a range of size $q \cdot (\frac{\Delta^2}{q} + O(\sqrt{\Delta})) + q = \Delta^2 + O(\Delta^{3/2})$. This is the resulting number of colors of step 3, whose running time is $O(\Delta^{3/2} \cdot \log \Delta)$ as described in Section \ref{arbLinial}.

Next we describe step 4, which reduces the number of colors from $\Delta^2 + O(\Delta^{3/2})$ to $\Delta^2 + 1$. This step is performed using an adaptation of a simple color reduction for $G$ to $G^2$. The simple color reduction for $G$ works as follows. Each vertex whose color is greater than all its neighbors colors (and greater than $\Delta(G) + 1$), selects a new color from $\{1,2, ... ,\Delta(G) + 1\}$ that is not used by any neighbor. By starting from a proper $(\Delta(G) + k)$-coloring, for a parameter $k > 1$, and repeating this for $k-1$ rounds, a proper $(\Delta(G) + 1)$-coloring is achieved. If message size is unbounded, this can be directly applied in $G^2$. However, each vertex needs to collect the colors of its $2$-distance neighborhood. This is in order to know which colors are available. Thus, $O(\Delta)$ rounds are required to simulate each original round. To perform this in the CONGEST setting efficiently, we improve this by invoking a binary search to find an available color. This is done as in Section \ref{Linal2Distance} and requires $O(\log \Delta)$ rounds. Consequently, an $O(\Delta^2 + O(\Delta^{3/2}))$-coloring of $G^2$ can be reduced into $(\Delta^2 + 1)$-coloring of $G^2$ within $O(\Delta^{3/2} \log \Delta)$ rounds. This completes the description of step 4. The result of invoking steps $1$ to $4$ is summarized in the next corollary.

\begin{col}
$(\Delta^2 + 1)$-coloring of $G^2$ can be computed within $O(\Delta^{3/2}\log \Delta + \log^* n)$ rounds in the CONGEST setting.
\end{col}

\section{Speedup technique for algorithms on $G^k$ in the CONGEST model}\label{SpeedupCongest}

In this section we devise a method which speeds up the running time of a wide class of algorithms for problems on $G^k$ in the CONGEST model. This includes algorithms for problems such as: MIS, maximal matching, edge coloring, vertex coloring, ruling set, cluster decomposition, etc. 
Our goal is reducing the amount of data passed in the network, by exchanging messages only half-the-way, compared to standard algorithm for $G^k$.
In the general case, in a problem for $k$-distance, a node may have up to $\Delta^k$ neighbours. Thus, previously-known solutions require $\Delta^{k-1}$ rounds for collecting information about $k$-hop neighbours in each step of an algorithm. This way, an algorithm with $f(\Delta,n)$ rounds for $G$ is translated into an algorithm with $f(\Delta^k,n) \cdot \Delta^{k-1}$ rounds for $G^k$. (In general, the running time may be even larger, but we focus on solutions in which each vertex makes a decision based on the current information of vertices in its $k$-hop-neighborhood, where each vertex holds $O(\log n)$ bits.)
On the other hand,  our new technique makes it possible to collect aggregated data from distance only $k/2$. As a result, the size of the collected data becomes $O(\Delta^{k/2})$. This allows us to obtain a running time of $f(\Delta^k,n) \cdot \Delta^{k/2-1}$ instead of $f(\Delta^k,n) \cdot \Delta^{k-1}$. Nowadays, various problems have algorithms with running time $f(\Delta,n) = O(\mbox{polylog}(\Delta) + \log^* n)$ for $G$. In such cases our technique provides a quadratic improvement for $G^k$, i.e., for distance-$k$ computations on $G$. Our technique is based on idempotent functions, described below.


\subsection{Idempotent functions}
An {\em idempotent function} is a function $f$ from a set $A$ to itself, such that for every $x \in A$, it holds that $f(f(x)) = f(x)$. For example, the following boolean functions from $\{0,1\}$ to $\{0,1\}$ are idempotent. 
$F_0(x) = x \mbox{ OR } 0$, 
$F_1(x) = x \mbox{ OR } 1$,
$G_0(x) = x \mbox{ AND } 0$,
$G_1(x) = x \mbox{ AND } 1$.
Another example, are functions $\hat{H}_t,\check{H}_t:N \rightarrow N$, where $t \in N$, defined as follows. 
$\hat{H}_t(x) = max(x,t)$, $\check{H}_t(x) = min(x,t)$.

A pair of functions $f(),g()$ is {\em commutative} if $f(g(x)) = g(f(x))$. A set of functions is a {\em commutative set} if any pair of functions in the set is commutative.\\
We define an {\em idempotent composition} as follows. A set $A$ with a commutative set of functions $f_1,f_2,...,f_k$ from $A$ to itself is an {\em idempotent composition}, if for any $q$  (not necessarily distinct) indices $j_1,j_2,...,j_q$,  in the range $[k]$, and $p \leq q$ distinct indices $i_1 \neq i_2,..., \neq i_p$, such that  $\{i_1,i_2,...,i_p\} =  \{j_1,j_2,...,j_q\}$ , it holds that
$$ f_{j_1}(f_{j_2}(...f_{j_q}(x))...)) = f_{i_1}(f_{i_2}(...f_{i_p}(x))...)). $$

For example, the set $\{0,1\}$ with the functions $F_0$, $F_1$ as defined above is an idempotent composition. 
Indeed, $F_0(F_0(...F_0(x)...)) = F_0(x) = x$, $F_1(F_1(...F_1(x)...)) = F_1(x) = 1$, and any composition of functions $F_0$ and $F_1$ equals $F_0(F_1(x)) = F_1(F_0(x)) = 1$.
\newline

\subsection{High level description of our technique for $G^k$} \label{secpowerg}

This section assumes $k$ is even. In the case $k$ is odd our technique cost another factor of $O(\Delta)$ of communication rounds per algorithm round. Hence it behaves as if the distance required is $k + 1$ in terms for CONGEST communication.
Our method for distance-$k$ computations consists of two stages.

In the first stage each vertex collects information from its $k/2$-hop-neighborhood. This is done by broadcasting in parallel from all vertices to distance $k/2$. Consequently, for any pair of vertices $u,v$ at distance $k$ one from another, there exists a vertex $w$ in the middle of a path between $u$ and $w$ that received the information of $u$ and $v$. Indeed, the distance between $u$ and $w$ is $k/2$, and between $v$ and $w$ it is $k/2$. Next, $w$ computes internally, for each $u$ in the $k/2$-hop-neighborhood of $w$, the available information for $u$ regarding its $k$-hop-neighbors. That is, the neighbors at distance $k$ from $u$ who are also in the $k/2$-neighborhood of $w$, Which, sometimes contains all of the vertices in the $k/2$-neighborhood of $w$ and sometimes only part of in case we would like to avoid double counting.

In the second stage, for all $v \in V$, all information computed for $v$ in its $k/2$-hop-neighborhood should be delivered to $v$. Recall that the information computed by vertices in the $k/2$-hop-neighborhood of $v$ is about the $k$-hop-neighborhood of $v$. But delivering the entire information in a straightforward way to $v$ requires up to $\Delta^k$ rounds. Indeed, vertices at distance $k/2$ from $v$ hold information of size up to $\Delta^{k/2}$, and the number of such vertices is up to $\Delta^{k/2}$. In order to reduce the amount of information that has to be passed, an aggregation function is used. Specifically, each vertex $v$ collects information in a convergecast manner. That is, each vertex at distance $k/2$ from $v$ sends information of size $\Delta^{k/2}$ to its neighbors. This requires $O(\Delta^{k/2})$ rounds. Now vertices at distance $k/2 - 1$ from $v$ have received information from their neighbors at distance $k/2$ from $v$. But instead of sending all this information to $v$ they perform an aggregation in which the information size shrinks to $\Delta^{k/2 - 1}$. In the following phase, the information size shrinks to $\Delta^{k/2 - 2}$,  etc. See Figure 1. 

\begin{figure}[htp]
    \centering
    \includegraphics[width=15cm]{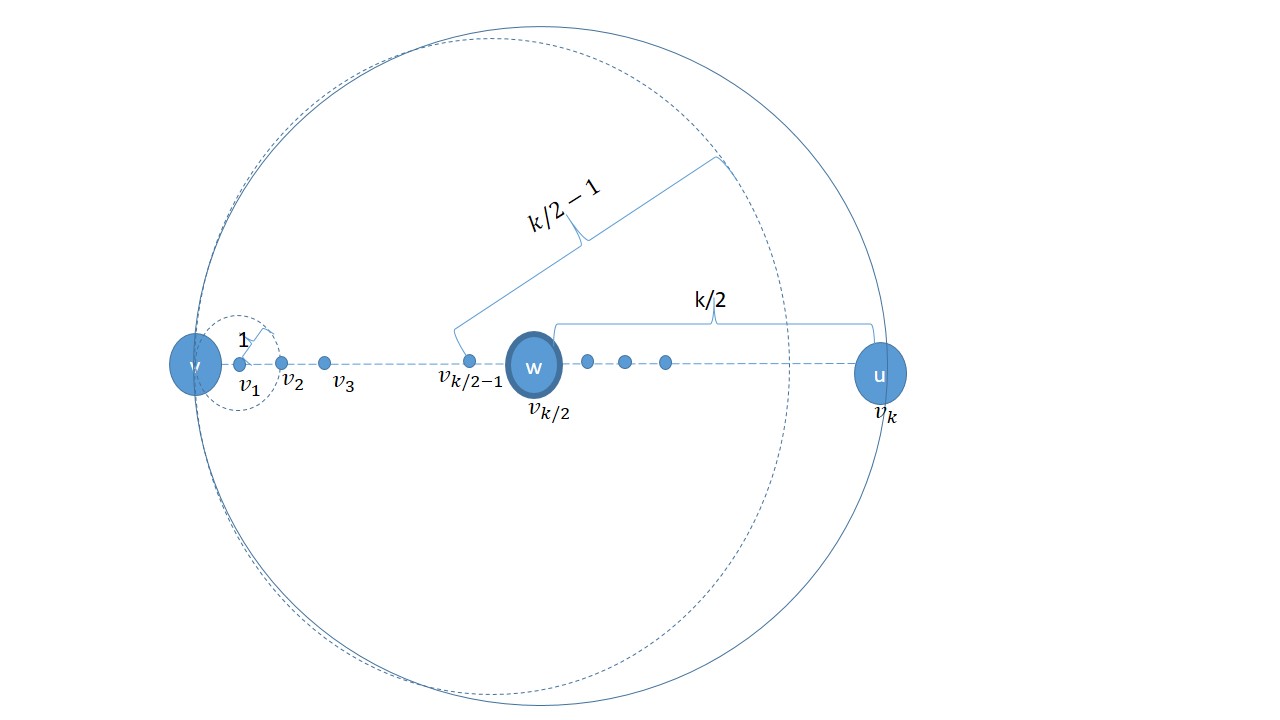}
    \caption{A vertex $w$ in the middle of a path between $v$ and $u$ collects information about its $k/2$-hop-neighborhood. Then a convergecast process is performed, in which balls of radius $k/2-i$, $i = 1,2,...,k/2$, around $v_{k/2 - i}$ are formed within $k/2$ stages. The balls contain aggregated information that after $k/2$ stages is about the $k$-hop-neighborhood of $v$. After stage $k/2$ this information resides in a ball of radius $0$ of $v$, i.e., in $v$ itself. }
    \label{fig:distance_figure}
\end{figure}

A basic building block of ours for the first stage is a procedure for broadcasting to distance $k/2$, initiated by all vertices in parallel. In this procedure, each message contains three elements: Node id of the originator of broadcast, the distance from the originator, and additional information. The part that contains the id and the distance is called {\em header}. In addition, each vertex $v$ maintains a local data structure $D_v$ that initially contains only the element of $v$. In the beginning of the broadcast procedure, each vertex sends a message with the contents of $D_v$ (in the additional information field of the message) to all its neighbors. During the execution of the broadcast procedure, if a node receives several messages with the same id of the originator, the node considers only one of these messages (which was the first to arrive; ties are broken arbitrarily) and ignores the others. In each round of the broadcast, each vertex updates its data structure $D_v$ with additional information it learnt, and sends it to its neighbors. Specifically the update of $D_v$ considers each received message and each element in the message with ID that the vertex has not received yet. The vertex $v$ adds each such element to $D_v$. Then the updated data structure $D_v$ is sent to $v$'s neighbors. The broadcast procedure continues this way until the distance reaches $k/2$. Each message with distance greater than $k/2$ from the origin is discarded. Thus the broadcast process stops after $k/2$ rounds from initiation. 


Once the broadcast stops, for any vertex $w$ at distance at most $k/2$ from $v$, the vertex $w$ computes locally an appropriate value for $v$, using aggregation functions. These functions are applied iteratively for each vertex in the $k/2$-hop-neighborhood of $w$.
Next, a convergecast process is performed for $k/2$ phases. Its goal is that each $v \in V$ learns the outcome of a series of function invocations that are applied to all vertices in the $k$-hop-neighborhood of $v$. (This is done with assistance of vertices at distance up to $k/2$ from $v$.) To this end, aggregation  functions $f: A  \rightarrow A, A \subseteq N_0$ are used. In certain cases $A$ is of size $2$, as in the example of deciding whether the value of $v$ appears in its $k$-hop-neighborhood. But in other cases $|A|$ may be larger. It has to satisfy $\log(|A|) = o(\Delta^{k/2})$, to provide a speed-up in our method. The number of aggregation functions is $|A|$, one per element of $A$. (This is because each aggregation result is an element of $A$, which may affect uniquely the next aggregation, thus an appropriate function for each element is needed.) In each phase $i = 1,2,...,k/2$ of the convergecast process each vertex $v \in V$ receives from its neighbors the results of the aggregation functions applied to each vertex $z$ in the $(k/2 + 2 - i)$-hop-neighborhood of $v$. (Note that these function applications hold results about distance $(k/2 + i)$ from $v$ at that stage.) For each such vertex $z$, there are up to $\Delta$ results of functions application, one per neighbor of $v$. Then for each $z'$ in the $(k/2 + 1 - i)$-hop-neighborhood of $v$, the vertex $v$ applies the aggregation function iteratively, up to $\Delta$ times, according to the number of results for $z'$ that $v$ received from its neighbors. 
This shrinks the amount of information by a factor of $\Delta$ for the next phase, where the radius of function application becomes smaller by $1$, and the radius of gained information becomes larger by $1$. After $k/2$ phases, each vertex $v \in V$ holds a single result of the applications of aggregation functions on all vertices in the $k$-hop-neighborhood of $v$.

A key property of the above process, which constitutes its correctness, is provided in the next lemma.

\begin{lem}
    After each phase $i = 1,2,...,k/2$, for each pair of vertices $v,u$ in $G$ at distance $k$ one from another, there exists a vertex $z$ at distance $k/2 - i$ from $v$ whose result of applications of the aggregation functions until this stage includes an invocation of an aggregation function on $u$. 
\end{lem}
\begin{proof}
    Let $(v = v_0,v_1,v_2,...,v_k = u)$ be a path of length $k$ between $v$ and $u$. The proof is by induction on $i$. \\
    {\bf Base ($i = 1$):} Just before the beginning of the convergecast stage, the vertex $v_{k/2}$ applies the aggregation function on its $k/2$-hop-neighborhood. Thus the result of these applications for $v$ includes a function invocation on $v_k = u$, who is at distance $k/2$ from $v_{k/2}$. This information is sent by $v_{k/2}$ to $v_{k/2 -1}$ in the first phase. Hence, there exists a vertex $z = v_{k/2 -1}$ at distance $k/2 - 1$ from $v$ whose result of applications of the aggregation functions until this stage includes an invocation on $u = v_k$.\\
    {\bf Step:} Suppose that after $i-1$ phases there exists a vertex $z$ at distance $k/2 - (i - 1)$ from $v$ whose result of applications of the aggregation functions until this stage includes an invocation on $u$. Moreover, suppose that $z = v_{k/2 - (i - 1) }$. In the next phase $i$, this result is sent to all neighbors of $z$, including $v_{k/2 - i}$. Hence, there exists a vertex $z = v_{k/2 -i}$ at distance $k/2 - i$ from $v$ whose result of applications of the aggregation functions until this stage includes an invocation on $u = v_k$.
\end{proof}

The above lemma shows that after $i$ phases of the convergecast stage the $(k/2 -i)$-hop-neighborhood of each vertex $v$ contains all needed information for $v$ about its $k$-hop neighborhood. Indeed for any vertex $u$ of distance up to $k$ from $v$ there is a vertex at distance $(k/2 -i)$ from $v$ who holds a result of a series of applications of the aggregation functions, one of which is an invocation for $v$ on $u$. Thus after $k/2$ phases each vertex $v \in V$ holds the result of invocations of the aggregation functions for all vertices in its $k$-hop-neighborhood. The number of rounds of each phase $i = 1,2,...,k/2$ is $O(\log |A| \cdot \Delta^{k/2 + 1 -i})$.

We illustrate this scheme by devising an algorithm for computing $k$-distance $O(\Delta^k)$-coloring. The algorithm generalizes $2$-distance algorithms as follows and based on the additive group coloring (AG) algorithm. The aggregation functions are the OR function, where $A = \{0,1\}$. Specifically: $F_0(x) = x \mbox{ OR } 0$ ; $F_1(x) = x \mbox{ OR } 1$. Each vertex $w \in V$ collects the current colors of its $k/2$-hop-neighborhood. For each vertex $v$ in the $k/2$-hop-neighborhood of $w$, the vertex $w$ computes whether the $b$ (We indicate a vertex color by the notation $\langle a , b \rangle$ when we compute additive group coloring) element of the color of $v$ equals to at least one of the $b$ elements of the other vertices in the $k/2$-hop-neighborhood of $w$. To this end, for each vertex $u$ in its $k/2$-hop-neighborhood, $w$ applies a function $F_i(x)$, where $i = 1$ if the $b$ values of $v$ equals to that of $u$, and $i = 0$ otherwise. Each invocation of  $F_i(x)$ by $w$, except the first invocation of a phase, is applied on the result of another invocation by $w$ in that phase. If this is the first invocation, then it is applied on $x = 0$. 
Once the vertices apply the function for their $k/2$-hop-neighbors, the resulting $\Delta^{k/2}$ values are sent to immediate neighbors. Consequently, each vertex $v \in V$ receives information of the $k/2$-hop-neighborhoods of its own neighbors. Note that vertices at distance up to $k/2-1$ from $v$ received information about vertices at distance up to $k$ from $v$. Now the same is repeated, but for distance $k/2 - 1$, rather than $k/2$. This continues for $k/2$ phases.  Finally, as a result of the function invocations and exchange of messages, each vertex $v \in V$ holds a single result, indicating whether there is another vertex in its $k$-hop-neighborhood with the same $b$ value. This allows to compute the next color of $v$, as in the AG algorithm. Thus, within $O(\Delta(G^k)) = O(\Delta^k)$ iteration of computing the next colors as described above, a proper distance-$k$ coloring with $O(\Delta^k)$ colors is obtained, from an initial $O(\Delta^{2k})$-coloring to distance $k$. The latter coloring can be obtained by generalizing $2$-distance-Linial's algorithm (Algorithm 1) to distance $k$. (We elaborate on this in Section \ref{nonidmptntf}.) Each iteration requires $O(\Delta^{k/2})$ rounds, and the overall running time is $O(\Delta^{k + k/2} + k \cdot \Delta^{k/2} \cdot \log \Delta \cdot \log^* n)$. We summarize this discussion with the next theorem.

\begin{thm}
Distance-$k$ coloring with $O(\Delta^k)$ colors can be computed within  $O(\Delta^{k + k/2} + k \cdot \Delta^{k/2} \cdot \log \Delta \cdot \log^* n)$ rounds in the CONGEST model.
\end{thm}

\noindent In Section \ref{nonidmptntf} we show how the running time for distance-$k$ coloring can be improved. 

We conclude the current section by improving the running time of simulating a round of an algorithm for $G$ in $G^k$. The improvement is from $O(\Delta^{k/2})$ to $O(\Delta^{k/2 - 1})$.
Consider the last phase of broadcast, when each vertex needs to receive information about  $(k/2 - 1)$-neighborhoods from each of its immediate neighbors. A $(k/2 - 1)$-neighborhood contains up to $\Delta^{k/2-1}$ vertices, each of which holds $O(\log n)$ bits of data for transmission. Thus, $O(\Delta^{k/2-1})$ rounds are needed in the last phase of broadcast to deliver the required information to each vertex, which then employs it to compute locally the information regarding its ${k/2}$-hop-neighborhood. Then, in the first phase of convergecast, instead of sending the computation of the entire $k/2$-hop-neighborhood, each neighbor receives information only about its $(k/2-1)$-hop-neighborhood. Indeed, in the second phase of the convergecast $(k/2-1)$-hop-neighborhoods are considered, so this information from the first phase is sufficient. This requires $O(\lceil \frac{\log |A|}{\log n}\rceil \cdot \Delta^{k/2 -1})$ rounds. The other phases require a smaller number of rounds. Thus the overall running time becomes $O(\lceil\frac{\log |A|}{\log n}\rceil \cdot \Delta^{k/2 -1})$. We summarize this in the next Corollary.

\begin{col}
    An $R$ round algorithm for $G$ in the CONGEST model that employs a set of commutative idempotent functions $f_1,f_2,...,f_t:A \rightarrow A$ can be transformed into an algorithm for $G^k$ in the CONGEST model, with running time $O(R \cdot \lceil \frac{\log |A|}{\log n} \rceil \cdot \Delta^{k/2 - 1}).$
\end{col}
We also obtain an improvement in the memory complexity, as shown in the next corollary.
\begin{col}\label{MemoryFootprint}
    Denote by $M$ the local memory of a vertex needed to complete an algorithm in  a straightforward way for $G^k$, $k \geq 2$, in the CONGEST model. The local memory required for an algorithm which uses our aggregation functions, such that $\log |A| = O(\log n)$, is bounded by $O(\frac{M}{\Delta^{k/2}}) \sim O(\sqrt{M})$. 
\end{col}
\begin{proof}
    Assume that an algorithm has aggregation functions with image $A$. Thus, the memory needed to store the required data resulting from functions applications is $\log |A|$ bits. This means that each vertex needs to hold at most $\log |A|$ bits times the number of vertices that are in distance $\frac{k}{2}$ from it. This is bounded by $O(\log |A| \cdot \Delta^{\frac{k}{2}})$. 
    The number of bits stored when using our technique instead of the naive solution is $O(\log |A| \cdot \Delta^{\frac{k}{2}})$ as opposed to $M = O(\log n \cdot \Delta^{k})$. (This is because in the naive solution every vertex must know its $k$-hop-neighborhood, which might contain $O(\Delta^k)$ vertices, and each of the vertices information can be stored using at $O(\log n)$ bits.) Thus, the local memory required for an algorithm that uses our aggregation functions is $O(\log |A| \cdot \Delta^{\frac{k}{2}}) = O(\frac{M}{\Delta^{k/2}}) \sim O(\sqrt{M})$.
\end{proof}
Memory considerations as in Corollary \ref{MemoryFootprint} are usually not analyzed in the CONGEST model. However, in real life applications that consider RAM limitations, our technique provides an improvement, which can make algorithm implementation much more efficient.

\subsection{Computations with non-idempotent functions} \label{nonidmptntf}
A function that is broadly used by us, and is not idempotent is the counting function. Specifically, given a value that a vertex holds, it needs to compute how many times this value appears in the $k$-hop-neighborhood of the vertex. It is applied on each vertex in the $k$-hop-neighborhood, and increments the result by one, each time it encounters the given value. (It does not change the result when it is applied on a different value.) Hence, the number of function applications may affect the result, and thus it is non-idempotent. Nevertheless, we are still interested in using such a function, since it appears in Linial's algorithm for computing number of polynomial intersections, and in arbdefective-coloring for computing the number of vertices of the $k$-hop-neighborhood with the same color.

There are two options to address this goal. The first option is analysing the solution with the possibility of over-counting. In this case the result is a correct upper bound on the desired solution, and the goal is analyzing how far it is from a solution that counts exactly. The second option is to obtain exact counting, using a more sophisticated construction, which we describe in the sequel. 

In the first option we perform an adaption from $G$ to $G^k$ as in the case of idempotent functions, described in Section {\ref{secpowerg}}. Now we analyze the maximum number of function invocations in an iteration. Note that each invocation corresponds to a unique path of length at most $k$ that starts from a vertex $v$, goes through a vertex $w$ in the middle of the path, and ends in a vertex in the $k/2$-hop-neighborhood of $w$. The number of such paths is bounded by $\Delta^k + \Delta^{k-1} +\Delta^{k-2}+....\Delta  = O(\Delta^k)$. Denote this bound by $d_p$. Now we can use $d_p$ instead of $\Delta^k$ in the computations for $k$-distance-Linial algorithm or arbdefective colorings to distance $k$. (Note that $\Delta$ is known to all vertices, so they can compute $d_p = \Delta^k + \Delta^{k-1} +\Delta^{k-2}+....\Delta$.) These computations will provide correct results, according to the same pigeonhole principle, as in the proof of the distance-$2$ variants, but the results now depend on $d_p$, rather than $\Delta^k$.

We illustrate this with an adaptation of Linial's algorithm to provide distance-$k$ coloring with $O(\Delta^{2k})$ colors, for a constant $k$. Let $d_p = O(\Delta^k)$, as described above. We compute a coloring with $O(d_p^2) = O(\Delta^{2k})$ colors. This is done by $O(\log^* n)$ stages, each of which perform a binary search for $O(\log \Delta^k) = O(\log \Delta)$ phases. The binary search is done using the aggregation counting function, to count the number of conflicts in distance $k$ of a polynomial $P$, for each half of its range. Even though this may cause over-counting, the aggregation functions are applied at most $d_p$ times per polynomial per phase. Thus using a field of size $q$, $d_p < q < 2d_p$, guarantees that there is going to be $t \in \{0,1,...,q - 1\}$, such that $P(t)$ does not intersect with any polynomial in its $k$-hop-neighborhood. Such value $t$ is going to be found by the binary search. This results in an $O(d_t^2)$-coloring in $O(\Delta^{k/2 -1} \cdot \log d_t \cdot \log^* n)$ time, i.e., $O(\Delta^{2k})$-coloring in $O(\Delta^{k/2 -1} \cdot \log \Delta \cdot \log^* n)$ time. We summarize this in the next corollary.

\begin{col}
    For a constant $k \geq 1$, we compute $O(\Delta^{2k})$-coloring of $G^k$ in $O(\Delta^{k/2 - 1} \cdot\log \Delta \cdot \log^* n)$ rounds in the CONGEST model.
\end{col}

While the above option is a simple solution, it provides an upper bound, which may be several times larger that the exact value. Also, for other functions, applying them more than the required number of times may not work. Thus we propose a second option that allows executing the function exactly the desired number of times. This requires a one-time preprocessing stage with $O(\Delta^{k - 1})$ rounds. Still, for algorithms with running time $R$ in $G$ of at least $\Delta^{1/2}$, this does not affect significantly the overall running time of the transformation for $G^k$. (Recall that the best currently-known deterministic algorithm for $O(\Delta)$-coloring and MIS have running times $\tilde{O}(\Delta^{1/2} + \log^* n)$ and $O(\Delta + \log^* n)$, respectively.)

In the preprocessing stage, each vertex $v \in V$ computes a BFS tree of height $k/2$ rooted at $v$. (Thus, the radius of the tree is also bounded by $k/2$.) This is the BFS tree of the $k/2$-hop-neighborhood of $v$. Then each vertex receives the BFS trees of its $k/2$-hop-neighborhood. Hence, information about $k$-hop-neighborhoods is obtained. The construction of BFS trees proceeds in $k/2$ phases. In each phase $i = 1,2,...,k/2$, BFS trees of height $i$ are constructed from BFS trees of height $i - 1$, as follows. A vertex receives the BFS trees of height $i - 1$ from its neighbors. It constructs locally a graph $G_v$, consisting of all vertices of all these trees, and  all edges of the trees. (Note that such a composition of trees may cause cycles.) Now a local BFS algorithm is executed on $G_v$, starting from $v$. Note that in the graph $G_v$ all edges containing $v$ connect it to roots of trees of height $i - 1$. Consequently, the BFS of $G_v$ results in a tree of height $i$. It is broadcasted to the neighbors of $v$ in phase $i$. Since a BFS tree of height $i$ may have up to $O(\Delta^i)$ vertices and edges, the running time of phase $i$ is $O(\Delta^i)$. The overall running time for computing BFS trees of height $k/2$ is $O(\Delta^{k/2 - 1})$.  This is because in the beginning of phase $k/2$, BFS trees of height $k/2 - 1$ are received, which requires $O(\Delta^{k/2-1})$ rounds. Then BFS trees of height $k$ are constructed locally, which completes the computation of the BFS tree.

Next, for each vertex $v \in V$ in parallel, the BFS tree that $v$ computed is sent to all vertices in the $k/2$-hop-neighborhood of $v$. This is done by parallel broadcasting within $k/2$ phases. Phase $1$ requires $O(\Delta^{k/2})$ rounds, phase $2$ requires $O(\Delta^{k/2 + 1})$ rounds, etc. Phase $k/2$ requires $O(\Delta^{k - 1})$ rounds, which is also the overall running time of the preprocessing stage.   

Next, we explain how these BFS trees are used to make sure that each function invocation in the convergecast stage of our algorithm is performed exactly once, for each vertex in the $k$-hop-neighborhood of $v$. Let $w \in V$ be a vertex that performs the convergast stage, and $v,u$ be vertices in the $k/2$-hop-neighborhood of $w$. The vertex $w$ has to decide whether to invoke a function of $v$ on $u$. Since $w$ holds the BFS tree of height $k/2$ of $v$, the vertex $w$ knows who are all the other vertices $w'$ that also considering at that stage whether to execute a function for $v$. Moreover, $w$ knows the set $W' = \{w' \ | \ w' \mbox{ is a vertex at distance } k/2 \mbox{ from } $v$ \\ \mbox{, and at distance at most } k/2 \mbox{ from } u \}$. This is because $w$ holds all BFS trees of its $k/2$-hop-neighborhood. Now, $w$ can find the vertex with the smallest ID in $W'$. If this is $w$ itself it applies the function, and does not apply it otherwise. This way, out of all vertices $w'$ that could have applied it, exactly one vertex does so.

It remains to address vertices $u$ at distance less than $k/2$ from $v$, for which there are no vertex $w$ at distance $k/2$ from $v$, such that $u$ is in the $k/2$-hop-neighborhood of $w$. This is addressed directly by $v$, using the BFS tree of $v$. 
Since the BFS tree of $w$ is also known to $v$, it is possible for vertex $v$ to decide to apply the function on such $u$ that are not in the BFS tree of $w$.

Now, our scheme that described in Section \ref{secpowerg} for computations on $G^k$ with idempotent functions can be extended to non-idempotent functions. The difference in the generalized scheme is that it starts with the preprocessing stage of computing BFS-trees of height $k/2$. In addition, during the main stages of the scheme, whenever a function has to be applied, a check is performed to ensure that it is applied exactly once for a pair $v,u$, as explained above. This gives rise to the following result.

\begin{thm} \label{distancealgfast}
    Suppose we are given a set $A$ of all possible inputs and outputs of aggregation functions, and an $R$-round algorithm for $G$ in the CONGEST setting, where each round consists of a series of invocations of commutative aggregation functions $f_1,f_2,...f_t:A \rightarrow A$. Then we can obtain an algorithm for the same problem on $G^k$ whose running time in the CONGEST model is $O(R \cdot \lceil \frac{\log |A|} {\log n} \rceil \cdot \Delta^{k/2 -1} + \Delta^{k-1})$.
\end{thm}


\subsection{Applications of the speedup technique for $G^k$}\label{speedup}
In this section we provide several application of Theorem \ref{distancealgfast}. Consider the $O(\Delta)$-coloring algorithm for $G$ on which the results of Section \ref{secfastcol} are built. It consists of three parts:\\
\noindent (1) Computing defective coloring. This part employs an aggregation function to compute number of intersections in Linial's algorithm.\\
\noindent (2) Computing arbdefective coloring. This part employs an aggregation function to compute the number of neighbors with the same color.\\
\noindent (3) Iterating over color classes and computing a proper coloring iteratively. This part employs an aggregation function that again computes the number of polynomial intersections, but this time according to the algorithm in Section \ref{arbLinial}.

In all these functions we can set $A = \{1,2,...,n\}$, since the results are always bounded by the number of vertices in the graph.
Since the algorithm for $G$ composed of these steps requires $\tilde{O}(\Delta^{1/2} + \log^* n)$ rounds \cite{BEG18}, it gives rise to an $\tilde{O}(k \cdot\Delta^{k-1} + k \cdot \Delta^{\frac{k}{2}} \log^* n)$-round algorithm for $G^k$, by Theorem \ref{distancealgfast}. (The factor of $k$ appears in the running time because the $O(\log \Delta)$ running time for binary searches in $G$ translates into $O(\log \Delta^k) = k \log \Delta$ phases in $G^k.)$ This is summarized in the next corollary.
\begin{col} \label{coldistalgc}
    For $k > 1$, an $O(\Delta^k)$-coloring of $G^k$ can be computed in $\tilde{O}(k \cdot\Delta^{k-1} + k \cdot \Delta^{\frac{k}{2}} \log^* n)$ rounds in the CONGEST model.
\end{col}

Next, we obtain an MIS algorithm for $G^k$. It starts with computing $O(\Delta^k)$-coloring of $G^k$. Then, for each color class $i = 1,2,...$, an iteration is performed, consisting of $k$ rounds. Specifically, all vertices of color $i$ perform broadcast to distance $k$. That is, each such vertex initializes a broadcast message with a counter that is initialized to $0$. Each time such a message is received, the counter is incremented by $1$. If a vertex receives several messages in parallel, it handles only one of them, and the others are discarded. If the counter of the message is smaller than $k$ it is sent to the immediate neighbors. Otherwise, it is discarded. Therefore, an iteration completes within $k$ rounds. Then vertices of color $i$ that initiated the broadcast decide to join the MIS, and vertices that received the messages decide not to join. Vertices that made decisions become inactive. Then the next iteration starts on the remaining active vertices, etc.  After $O(\Delta^k)$ such iterations, the algorithm terminates. 

\begin{thm}
    MIS of $G^k$ can be computed within $\tilde{O}(k \cdot \Delta^k + k \cdot \Delta^{k-1} \cdot \log^* n)$ rounds in CONGEST model.
\end{thm}
\begin{proof}
An $O(\Delta^k)$-coloring of $G^k$ can be computed in $\tilde{O}(k \Delta^{k-1} \log^* n)$ rounds, by Corollary \ref{coldistalgc}. Next, if a pair of vertices of the same color join the set, they are at distance at least $k+1$ one from another in $G$. For a pair of vertices $u,v$ of distinct colors that join the set, one of them joins before the other. Assume without loss of generality that this is $u$. After $u$ joins, all vertices at distance at most $k$ from $u$ decide not to join, as a result of the broadcasting of $u$. The vertices become inactive, but $v$ that joins the set in a later stage is active at that time. Hence $u$ and $v$ are at distance at least $k+1$ one from another. Consider a vertex $v'$ that did not join the set. Let $j$ be the color of $v'$. Then in iteration $j$ the vertex $v'$ was inactive. This means that in an earlier iteration it received a message of broadcast with a counter value between $1$ and $k$. Thus $v'$ is at distance at most $k$ from some vertex in the independent set. Thus the result of the algorithm invocation is an MIS of $G^k$. The number of iteration is $O(\Delta^k)$, each requires $k$ rounds, which results in overall running time   $\tilde{O}(k \cdot \Delta^k + \Delta^{k-1} \log^* n)$.
\end{proof}

\noindent {\large \bf Conclusion \\}
Our results demonstrate the usefulness of the meet-in-the-middle approach in distributed graph algorithms for the CONGEST setting. Our techniques give rise to a quadratic improvement for simulating a round of certain algorithms for $G$ on $G^k$. These are algorithms that can be expressed using commutative and idempotent operations. In the case $k = 2$, the techniques provide an exponential improvement in some cases. The improvement is most significant when the original algorithm has very low number of rounds on $G$, but the standard extension to $G^2$ incurs a penalty of $\Delta$ rounds. In such cases, our approach may provide an improved algorithm with just $\mbox{polylog}$ dependency on $\Delta$, if the original algorithm has certain properties. The entire transformation of an algorithm from $G$ to $G^k$ also depends on the number of rounds of the algorithm for $G$. Therefore, the best results are obtained when this number is small. Thus an improvement of an algorithm for $G$ would give an improvement to the extension for $G^k$. In certain algorithms a BFS tree has to be constructed for the extension to $G^k$. This affects the number of rounds as well. It would be interesting to find a solution that does not employs BFS, and as a consequence improves running time. It would be also interesting to discover simpler implementations of the techniques. Recently, Maus \cite{M21} obtained simplified algorithm for various coloring problems. Some of these problems with their (not necessarily simplified) solutions are used as building blocks in our paper. It is an open question whether our entire construction can be simplified too. \\

\noindent {\large \bf Acknowledgments \\}
The authors are grateful to Michael Elkin for fruitful discussions.
\bibliography{bib2doi}

\appendix
\noindent {\large \bf APPENDIX}
\section{$O(\Delta^2)$-Coloring of $G^2$ in time $O(\Delta^2 + \log^* n$)}\label{AG2D}

In this section we devise a distributed algorithm whose input is an $O(\Delta^4)$-coloring of $G^2$, and it computes an $O(\Delta^2)$-coloring of $G^2$. This is equivalent to 2-distance $O(\Delta^2)$-coloring of $G$.  The algorithm in the current section employs  1-bit messages.
The input for our algorithm is a graph $G$ with a given proper 2-distance $O(\Delta^4)$-coloring $\phi$. (We explain how to compute such an $O(\Delta^4)$-coloring in Section \ref{Linal2Distance}.) 
Let $q = O(\Delta^2)$ be the smallest prime number, such that the number of colors in the coloring is bounded by $q^2$, and $q > 2\Delta^2$.
Denote $\phi(v) = \langle a , b \rangle$, such that $a,b \in O(\Delta^2)$. The coloring remains proper during all stages of our algorithm and eventually the colors of all nodes obtain the form of $\phi(v) = \langle 0 , b \rangle$. This results in an $O(\Delta^2)$-coloring of the entire graph. The algorithm proceeds in phases as follows. In each phase we consider the value $b$ of each vertex color and its 2-distance neighbors. If a vertex has a value of $b$ that is distinct from the $b$ values of all its 2-hop neighbors it will change its color to $\phi(v) = \langle 0 , b \rangle$. Otherwise, the $b$ value of $v$ equals to at least one $b$ value of a vertex in the 2-hop-neighborhood of $v$. In this case, we say there is a {\em conflict}. Then, the vertex $v$ will update $\phi(v) = \langle a , b + a \mbox{ mod } q \rangle$. This completes the description of a phase.

The described phase is a generalization of AG (Additive group) coloring \cite{BEG18} for supporting 2-distance coloring. However, instead of a straightforward generalization that invokes the algorithm of \cite{BEG18} on $G^2$, we devise a bit-efficient procedure. The main idea is that in order to support 2-distance proper coloring it is not necessary for each node to know which colors its neighbours chose. The only required information for a vertex is whether there is a conflict with its $b$-coordinate in its 2-hop-neighborhood. In order to know this information each vertex $v \in V$ sends an indication to its 1-hop neighbors about the color it intends to choose. If this color is not targeted by any vertex in its 2-hop neighborhood, all 1-hop neighbors of $v$ can check that there are no conflicts and notify $v$ accordingly. If there are no conflicts then $v$ finalizes the intended color. Otherwise, at least one 1-hop-neighbor of $v$ discovers a conflict. It notifies $v$, which discards the proposed color and proposes a new one in the next phase, as explained above. We devise a bit-efficient algorithm, where we use each node's neighbors to aggregate information about conflicts in order to decide the color of the next phase. In a bandwidth restricted model (i.e. CONGEST or one-bit) the running time of our 2-distance-AG-algorithm is faster by a factor of $\Delta$ compared to simply running AG coloring \cite{BEG18} on $G^2$. Specifically, when running the original AG algorithm on $G^2$ in a bandwidth restricted model \cite{BEG18}, each node has to be aware of the $O(\Delta)$ colors passed by its 2-distance neighbors. In order to send all $\Delta$ colors, $\Delta$ rounds are required for the message passing of each step in the original AG algorithm. This means that the running time of AG as a black-box on $G^2$ is $O(\Delta^3)$ rounds. This is because running AG on $G^2$ requires $O(\Delta^2)$ phases, each of which may have to transmit $O(\Delta)$ colors over each edge. Therefore, the running time is $O(\Delta^2) \cdot \Delta = O(\Delta^3)$. On the other hand, our new algorithm improves the running time to $O(\Delta^2)$ rounds. \\

\noindent {\bf A.1 \ \ High level description \\}
The algorithm is a generalization of the 1-bit AG algorithm described in \cite{BEG18}. But instead of sending a message to all of the 2-distance neighbors, which is similar to a round of the original AG algorithm for $G^2$, each node only needs to send the decision to its immediate neighbors, 
which remember the decisions and aggregate them. 
The immediate neighbors of a vertex $v \in V$ receive decisions from their own neighbors, and thus hold all decisions of the 2-hop-neighborhood of $v$. Each immediate neighbor of $v$ checks for conflicts, and notifies $v$ whether at least one conflict occurred. If $v$ receives such a notification from at least one of its neighbors, it deduces that there is a conflict in its 2-hop-neighborhood. Otherwise, it has no conflicts in its $2$-hop-neighborhood.
This way, based on the received notification, a node can choose a proper color without knowing which colors its 2-distance neighbors chose. See pseudocode below. 

\noindent {\bf A.2 \ \ Pseudocode\\}
(The execution starts from Algorithm 6, which invokes Algorithms 4 - 5 as sub-routines.) \\
\begin{algorithm}
	\begin{algorithmic}[1]

    \STATE /* Check whether the color $<a_0, b_0>$ conflicts with at least one color from $\langle a_1, b_1 \rangle$, .. ,$\langle a_d, b_d \rangle$ */
    
	\IF {$\forall i \in {1,2,..d}$ it holds that $b_0 \neq b_i$}
	    \STATE return false
	\ELSE
	    \STATE return true
	\ENDIF
    \end{algorithmic}
    \caption{Procedure conflict($\langle a_0, b_0 \rangle$, { $\langle a_1, b_1 \rangle$) .. $\langle a_d, b_d \rangle$ ) }}
\end{algorithm}

\begin{algorithm}
\begin{algorithmic}[1]

\STATE /* Compute the next color $\phi_{i + 1}(v)$, based on the current color $\phi_i(v) = <a,b>$ and a set of boolean variables $R_1,....,R_d$, that represent conflicts. \\
Return true if the next color has an update in the $b$-value, otherwise return false */

\IF {all of $R_1,R_2,...,R_d$ equal false}
    \STATE $\phi_{i+1}(v) = \langle 0, b \rangle$
    
    \STATE return false
\ELSE
	\STATE $\phi_{i+1}(v) = \langle a, a + b \mbox{ mod } q\rangle$
	
	\STATE return true
\ENDIF

\caption{Procedure next-color($\phi_{i}(v)= \langle a, b \rangle, \{R_1,...,R_d\}$)}
\end{algorithmic}
\end{algorithm}
\begin{algorithm}
	\begin{algorithmic}[1]
		
		\STATE /* Initially we are given a graph with proper $k$ coloring $\phi$. Let $q$ denote the smallest prime number such that $\sqrt{k} \leq q $ and $2\Delta^2 < q$ */

        \STATE Let $d$ denote the number of neighbors of $v$

		\STATE Denote $\phi(v) = \langle a , b \rangle$,  where $a, b \leq q$
		
		\STATE $C_{\Gamma} = \{ \phi(u) \ | \ u \in \Gamma(v) \}$
		
		\WHILE {$\phi(v) \neq \langle 0, b \rangle$}
		
		\STATE $C_{\Gamma} = C_{\Gamma} \cup \phi(v)$

		\FORALL{$u_i \in \Gamma(v)$ in parallel}
		    
		    \STATE $R_i$ = conflict$(\phi(u_i), C_{\Gamma} \setminus \phi(u_i)) $  \ \ \ \ \ /* Invocation of Algorithm 4 */
		    
		    \STATE send $R_i$ to $u_i$
		
		\ENDFOR
		
		\STATE Receive a message from each of the neighbors  $u_1,u_2,...,u_d$, of $v$, and store them in $R'_1,R'_2,....,R'_d$, respectively.

        \STATE $S_v$ = next-color ($\phi(v), \{R'_1,R'_2,....,R'_d$\}) \ \ \ \ \ \  /* Invocation of Algorithm 5 */
		
		\STATE send $S_v$ to all neighbors
		
		\STATE Receive all neighbor messages $S_u$ from line 13 /* $1$ bit each */
		
		\STATE $\forall u \in \Gamma(v)$ compute $\phi_{i+1}(u)$ using next-color$(\phi(u),\{S_u\})$. \ \ \ \ \ \ \ /* Invocation of Algorithm 5 */

		\ENDWHILE
	\end{algorithmic}
	\caption{2-Distance-AG}
\end{algorithm}
\clearpage
Next, we analyze the correctness of the algorithm. Lines 6 - 15 are called a {\em phase} of Algorithm 6.
\begin{lem}\label{aground}
	Suppose we are given a graph $G=(V, E)$ with a proper distance-$2$ coloring $\varphi$ using $O(\Delta^4)$ colors. After each phase of Algorithm 2-Distance-AG the coloring remains a proper distance-$2$ coloring. 
\end{lem}
\begin{proof}
 Let $u,v$ be two vertices of distance at most $2$ one from another. First, assume they are exactly at distance $2$ one from another. In this case there is a common neighbor $w$ of $u,v$, i.e., $(u, w) \in E$ and $(w,v) \in E$. Denote the colors of $u,v$ in the beginning of a phase by $\varphi(u) = \langle a_u,b_u \rangle$, $\varphi(v) = \langle a_v,b_v \rangle$. Denote the colors of $u,v$ in the end of the phase by $\varphi'(u),\varphi'(v)$. Assuming that the coloring is proper in the beginning of a phase, i.e., $\varphi(u) \neq \varphi(v)$, we next prove that it is proper in the end of the phase as well. We divide the proof into several cases.\\ 
 {\bf Case 1:} $b_u = b_v$. In this case, $w$ discovers in line 8 of Algorithm 6 that there is a conflict between $u$ and $v$. Then $w$ notifies $u$ and $v$ about the conflict. Consequently, $a_u$ and $a_v$ do not change in this phase. Since in the beginning of the phase it holds that $\varphi(u) \neq \varphi(v), b_u = b_v$, it follows that $a_u \neq a_v$ in the beginning of the phase. This also holds in the end of the phase, and thus $\varphi'(u) \neq \varphi'(v)$. \\
 {\bf Case 2:} $b_u \neq b_v$ and ($a_u = 0$ or $a_v = 0$ or both). If both $a_u,a_v$ equal $0$ in the beginning of the phase, then the colors of $u,v$ do not change during the phase, and thus $\varphi'(u) \neq \varphi'(v)$.  Otherwise, exactly one of $a_u,a_v$ equals $0$ in the beginning of the phase. If this is the case in the end of the phase, then $\varphi'(u) \neq \varphi'(v)$. Otherwise, $b_u,b_v$ do not change during the phase, and since $b_u \neq b_v$, it holds that $\varphi'(u) \neq \varphi'(v)$ as well.
 {\bf Case 3:} $b_u \neq b_v$ and $a_u \neq 0$ and $a_v \neq 0$. If in the end of the phase both $a_u \neq 0$ and $a_v \neq 0$, it means that $a_u,a_v$ did not change. If $a_u = a_v$ then $b_u \neq b_v$ in the beginning of the phase, and in the end of the phase it holds that $\varphi'(u) = \langle a_u, (b_u + a_u) \mbox{ mod } q \rangle$, $\varphi'(v) = \langle a_v, (b_v + a_v) \mbox{ mod } q \rangle$, and thus $\varphi'(u) \neq \varphi'(v)$. If $a_u \neq a_v$, and they did not change during the phase, then $\varphi'(u) \neq \varphi'(v)$ as well.  If exactly one of $a_u,a_v$ became $0$ in the end of the phase, again $\varphi'(u) \neq \varphi'(v)$. Finally, if both $a_u,a_v$ became $0$, then $b_u,b_v$ did not change, and again $b_u \neq b_v$. \\
 Hence, in all possible cases, $\varphi'(u) \neq \varphi'(v)$. \\
 If $u,v$ are direct neighbors, rather than at distance $2$ one from another, $u$ and $v$ discover directly whether a conflict occurs, and the same result holds.	
\end{proof}	
Next, we analyze the number of rounds required for all vertices to obtain an $O(\Delta^2)$-coloring.
\begin{lem}\label{agroundsnumber}
	For any prime $q > 2 \cdot \Delta^2$, after $q$ rounds all nodes will change their color to the form $\phi(v) = \langle 0, b \rangle$.
\end{lem}
\begin{proof}
    By Lemma \ref{aground}, after each phase, the coloring remains proper.  For $\langle a_u,b_u \rangle \neq \langle a_v,b_v\rangle$, the equation $a_u \cdot x + b_u = a_v \cdot x + b_v$ in $Z_q$ has exactly one solution, since $q$ is prime. Thus, if the vertices $v$ and $u$ have a conflict of their $b$-value in phase $i$, then they can be in conflict only one more time in the next $q-1$ phases, when one of their $a$-values change to $0$. (This happens, without loss of generality, when the color of $v$ is in the form $\langle 0 , b \rangle$ and the color of $u$ is in the form $\langle a , b \rangle, a > 0$.) We choose $q > 2 \cdot \Delta^2$. Then every vertex can have at most 2 conflicts with each of its 2-hop neighbours within $q$ rounds. From the pigeonhole principle, eventually there is going to be a phase for each vertex, where its $b$-value has no conflicts. This is because the number of vertices in a $2$-hop-neighborhood is at most $\Delta^2$, and with each of them there are at most $2$ conflicts during $q$ rounds. In addition, each phase (Lines 6-15 in Algorithm 6) require $O(1)$ rounds. (This holds even if 1-bit messages are sent per edge per round.) Therefore, the algorithm total running time is $O(\Delta^2)$ rounds.
\end{proof}

\end{document}